\documentclass[11pt,a4paper]{article}
\usepackage{amsfonts}
\usepackage[dvips]{graphicx}
\usepackage{amsmath}
\usepackage{mathrsfs}
\usepackage{makeidx}
\usepackage{xypic}
\usepackage[nottoc]{tocbibind}
\usepackage{pstricks}
\usepackage{bbm}
\usepackage[arrow, matrix, curve]{xy}
\usepackage{amsthm}
\usepackage{amssymb}
\usepackage{epic}
\usepackage{eepic}
\usepackage{epsfig}
\usepackage[numbers,square,sort&compress]{natbib}

\xymatrixcolsep{0.5cm}
\xymatrixrowsep{0.5cm}
\newtheorem{theorem}{Theorem}[section]
\newtheorem{proposition}{Proposition}[section]
\newtheorem{lemma}{Lemma}[section]

\newtheorem{corollary}{Corollary}[section]

\newcommand{\beqa}{\begin{eqnarray}}
\newcommand{\eeqa}{\end{eqnarray}}
\newcommand{\rf}[1]{(\ref{#1})}

\newcommand{\la}{\lambda}

\def\det{\operatorname{det}}
\numberwithin{equation}{section}

\begin{document}

\begin{flushright}
LPENSL - 2014\\
IMB - 2014
\end{flushright}

\bigskip \vspace{5pt}

\begin{center}
{\Large Open spin chains with generic integrable boundaries:%
\\[0pt]
\vspace{4pt} Baxter equation and Bethe ansatz completeness from SOV}

\vspace{50pt}

{\small \textbf{N.~Kitanine}\footnote[1]{IMB, UMR 5584 du CNRS, Universit\'e
de Bourgogne, France, Nikolai.Kitanine@u-bourgogne.fr},~~ \textbf{J.-M.~Maillet}\footnote[2]{Laboratoire de Physique, UMR 5672 du CNRS, ENS Lyon, Lyon 1 University, France, maillet@ens-lyon.fr},~~ \textbf{G. Niccoli}\footnote[3]{Laboratoire de Physique, UMR 5672 du CNRS, ENS Lyon, Lyon 1 University, France, giuliano.niccoli@ens-lyon.fr} }
\vspace{50pt}

\centerline{\bf Abstract} \vspace{1cm} 
\parbox{15cm}{\small 
We solve the longstanding problem to define a functional characterization of the spectrum of the transfer matrix associated to the most general spin-1/2 representations of the 6-vertex reflection algebra for  general inhomogeneous chains. The corresponding homogeneous limit reproduces the spectrum of the Hamiltonian of the spin-1/2 open XXZ and XXX quantum chains with the most general integrable boundaries. The spectrum is characterized by a second order finite difference functional equation of Baxter type with an inhomogeneous term which vanishes only for some special but yet interesting non-diagonal boundary conditions. This  functional equation  is shown to be equivalent to the known separation of variable  (SOV) representation hence proving that it defines a complete characterization of the transfer matrix spectrum. The polynomial character of the Q-function allows us then to show that a finite system of equations of generalized Bethe type can be similarly used to describe the complete transfer matrix spectrum.}
\end{center}

\newpage
\tableofcontents \newpage

\section{Introduction}

The functional characterization of the complete 
transfer matrix spectrum associated to the most general spin-1/2
representations of the 6-vertex reflection algebra on general inhomogeneous
chains is a longstanding open problem. It has attracted much attention in the framework of quantum integrability producing so far only partial results. The interest in the solution of this problem is at
least twofold. On the one hand, the quantum integrable system associated to
the limit of the homogeneous chain, i.e. the open spin-1/2 XXZ quantum chain 
with arbitrary boundary magnetic fields, is an interesting physical
quantum model. It appears, in particular, in the context of out-of-equilibrium physics
ranging from the relaxation behavior of some classical stochastic processes,
as the asymmetric simple exclusion processes \cite{EssD05,EssD06}, to the
transport properties of the quantum spin systems \cite{SirPA09,Pro11}. Their  solution can lead to non-perturbative physical
results and a complete and manageable functional characterization of their
 spectrum represents the first fundamental steps in this
direction. On the other hand, it is important to remark that the 
analysis of the spectral problem of these integrable quantum models turned
out to be quite involved by standard Bethe ansatz \cite{Bet31,FadST79}
techniques. Therefore, these quantum models are natural  
laboratories where to define alternative non-perturbative approach to their
exact solution. Indeed, the
algebraic Bethe ansatz, introduced for open systems by Sklyanin \cite{Skl88}
based on the Cherednik's reflection equation \cite{Che84}, in the case of
open XXZ quantum spin chains can be applied directly only in the case of
parallel z-oriented boundary magnetic fields. Under these special boundary
conditions the spectrum is naturally described by a finite system of Bethe
ansatz equations. Moreover the dynamics of such systems can be studied by  exact
computation of correlation functions \cite{KitKMNST07,KitKMNST08}, derived
from a generalisation of the method introduced in \cite{KitMT99,KitMT00,MaiT00} for periodic
spin chains.

Introducing a Baxter $T$-$Q$ equation, Nepomechie \cite{Nep02,Nep04}
first succeeded to describe the spectrum of the XXZ spin chain with
non-diagonal boundary terms 
in the case of an anisotropy parameter associated to the  roots of unity; furthermore, the result was obtained there only if
the boundary terms satisfied a very particular constraint relating the
magnetic fields on the two boundaries. This last constraint was also
used in \cite{CaoLSW03} to introduce a generalized algebraic Bethe
ansatz approach to this problem inspired by papers of Baxter \cite{Bax72,Bax72a}
and of Faddeev and Takhtadjan \cite{FadT79} on the XYZ spin chain. This
method has led to the first construction of the eigenstates of the XXZ spin
chain with non z-oriented boundary magnetic fields and this construction has
been obtained for a general anisotropy parameter, i.e.,  not restricted to the roots of unity
cases\footnote{%
Different methods leading to Bethe ansatz equations have been also proposed
under the same boundary conditions by using the framework  of the Temperley-Lieb algebra
 in \cite{deGP04, NicRd05} and by making a combined use of
coordinate Bethe ansatz and matrix ansatz in \cite{CraRS10, CraR12}.}. In 
\cite{YanZ07} a different version of this technique based on
the vertex-IRF transformation was proposed  but in fact it required one additional
constraint on the boundary parameters to work. It is worth mentioning that
even if these constrained boundary conditions are satisfied and generalized
Bethe ansatz method gives a possibility to go beyond the spectrum,
as it was done for the diagonal boundary conditions, no representation  for the scalar product of Bethe vectors\footnote{%
Some partial results in this direction were achieved in \cite{FilK11} but
only in the special case of double boundary constrains introduced in \cite%
{YanZ07}.} and hence for the
 correlation functions were obtained.

This spectral problem in the most general setting has then been also addressed by
other approaches. It is worth mentioning a new functional method leading to 
nested Bethe ansatz equations  presented in \cite{Gal08} for the
eigenvalue characterization and analogous to those previously introduced in 
\cite{MurN05} by a generalized $T$-$Q$ formalism. The eigenstate
construction has been considered in these general settings in \cite%
{BasK05a,Bas06} by developing the so-called $q$-Onsager algebra formalism.
In this last case the characterization of the spectrum is given by
classifying the roots of some characteristic polynomials. More recently, in 
\cite{CaoYSW13-4}  an ansatz $T$-$Q$ functional
equations for the spin chains with non-diagonal boundaries has been proposed\footnote{%
See also the papers \cite{CaoYSW13-1,CaoYSW13-2,CaoYSW13-3} for the
application of the same method to different models.}.

It is extremely important to remark that in general all methods based on Bethe ansatz (or
 generalized Bethe ansatz) are 
lacking proofs of the completeness of the spectrum  and in most
cases the only evidences of completeness are based on numerical checks
for short length chains. This is the case for the XXZ
chain with non-diagonal boundary matrices with the boundary constraint
for which the completeness of the spectrum description by the associated
system of Bethe ansatz equations has been studied numerically  \cite{Nep-R-2003,Nep-R-2003add}. In the case of the XXZ chain with completely general
non-diagonal boundary matrices some numerical analysis is also presented in 
\cite{CaoYSW13-4}. Further numerical analysis have been developed in a
much simpler case of the isotropic XXX spin chain where the most general boundary
conditions can be always reduced by using the $SU(2)$ symmetry
to one diagonal and one non-diagonal boundary matrices. For the  XXX chains
the ansatz introduced in \cite{CaoYSW13-2} was also applied and the completeness  of the Bethe ansatz spectrum was checked numerically
\cite{CaoJYW2013}. It is also important to mention a simplified ansatz proposed by
Nepomechie  based on  a
standard second order difference functional $T$-$Q$ equation with an
additional inhomogeneous term. The completeness of the Bethe ansatz
spectrum  has been verified numerically for small XXX chains in 
\cite{Nep-2013} while in \cite{Nep-W-2013} the problem of the description
of some thermodynamical properties has been addressed. 

These interesting developments attracted our attention in connection to the quantum separation of
variables (SOV) method pioneered by Sklyanin \cite{Skl85,Skl92}. The first analysis of the spin chain in the classical limit  from this point of view was performed in \cite{Skl89a,Skl89b}. This alternative approach 
allows to obtain (mainly by construction) the complete set of eigenvalues
and eigenvectors of quantum integrable systems. In particular, it 
was recently developed \cite%
{NicT10,N10-1,N10-2,GroMN12,GroN12,Nic12b,Nic13a,N13-1,Nic13b,Nic13c,Fald-KN13,FaldN13,GroMN13}
for a large variety of quantum models not solvable by algebraic
Bethe ansatz.  Moreover it has been shown first in \cite{GroMN12} that once the SOV spectrum characterization is
achieved manageable and rather universal determinant formulae can be derived for
matrix elements of local operators between transfer matrix eigenstates. In
particular, this SOV method was first developed in \cite{Nic12b} for the spin-1/2 representations of the 6-vertex reflection algebra
with quite general non-diagonal boundaries and then generalized to the most
general boundaries in \cite{Fald-KN13}. There, it gives the complete
spectrum (eigenvalues and eigenstates) and already allows to compute matrix
elements of some local operators within this most general boundary
framework. However, it is important to remark that this SOV characterization
of the spectrum 
is somehow unusual in comparison to more traditional characterizations like those obtained from Bethe
ansatz techniques. More precisely, the spectrum is described not in
terms of the set of solutions to a standard system of Bethe ansatz equations
but is given in terms of  sets of solutions to a characteristic
system of $\mathsf{N}$ quadratic equations in $\mathsf{N}$ unknowns, $\mathsf{%
N}$ being the number of sites of the chain. While the clear
advantage of this SOV characterization is that it permits to 
characterize completely the spectrum without introducing any ansatz one has to stress that the classification of the sets of solutions
of the SOV system of quadratic equations represents a new problem in quantum
integrability which  requires a deeper and systematic analysis.

The aim of the present article is to show that the SOV analysis of the transfer matrix
spectrum associated to the most general spin-1/2 representations of the
6-vertex reflection algebra on general inhomogeneous chains is strictly
equivalent to a system of \textit{generalized} Bethe ansatz equations.
This ensures that this system of Bethe equations characterizes automatically
the entire spectrum of the transfer matrix. More in
detail, we prove that the SOV characterization is equivalent to a second
order finite difference functional equation of Baxter type:
\begin{equation}
\tau (\lambda )Q(\lambda )=\mathbf{A}(\lambda )Q(\lambda -\eta )+\mathbf{A} 
(-\lambda )Q(\lambda +\eta )+F(\lambda ),
\end{equation}
which contains an inhomogeneous term $F(\lambda )$ independent on the $\tau $ and $Q$-functions and entirely fixed by the boundary parameters. It 
vanishes only for some special but yet interesting non-diagonal boundary
conditions (corresponding to the  boundary constraints mentioned above). One central requirement in our construction of this functional
characterization is the polynomial character of the $Q$-function. Indeed, it is
this requirement that allows then to show that a finite system of equations
of generalized Bethe ansatz type can be used to describe the complete
transfer matrix spectrum. Note that similar
results on the reformulation of the SOV spectrum characterization in terms
of functional $T$-$Q$ equations with $Q$-function solutions in a well defined
model dependent set of polynomials were previously derived  
\cite{N10-1,N10-2,GroN12} for the cases of transfer matrices
associated to cyclic representations of the Yang-Baxter algebra. The analysis presented here is also interesting as it introduces the main tools to generalize this type of reformulation to other classes of integrable quantum models.
The article is organized as follows. In Section 2 we set the main notations and we recall the main results of previous papers on SOV necessary for our purposes. Section 3 contains the main results of the paper with the reformulation of the SOV characterization of the transfer matrix spectrum in terms of the inhomogeneous Baxter functional equation and the associated finite system of generalized Bethe ansatz equations. In Section 4 we define the boundary conditions for which the inhomogeneity in the Baxter equation identically vanishes, in this way deriving the completeness of standard Bethe ansatz equations. There, we moreover derive the SOV spectrum functional reformulation for the remaining boundary conditions compatibles with homogeneous Baxter equations. Section 5 contains the description of a set of discrete transformations which leave unchanged the SOV characterization of the spectrum in this way proving the isospectrality of the transformed transfer matrices. These symmetries are used to find equivalent functional equation characterizations of the spectrum which allow to generalize the results described in Section 3 and 4. In Section 6 we present the SOV characterization of the spectrum for the rational 6-vertex representation of the reflection algebra and the reformulation of the spectrum by inhomogeneous Baxter equation. Finally, in Section 7, we present a comparison with the known numerical results in the literature for both the XXZ and XXX chains; the evidenced compatibility suggests that even in the homogenous chains our spectrum description is still complete.

\section{Separation of variable for spin-1/2 representations of the reflection
algebra}

\subsection{Spin-1/2 representations of the reflection algebra and open XXZ
quantum chain}

The representation theory of the reflection algebra 
can be studied in terms of the solutions $\mathcal{U}(\lambda )$ (monodromy
matrices) of the following reflection equation:%
\begin{equation}
R_{12}(\lambda -\mu )\,\mathcal{U}_{1}(\lambda )\,R_{21}(\lambda +\mu -\eta
)\,\mathcal{U}_{2}(\mu )=\mathcal{U}_{2}(\mu )\,R_{12}(\lambda +\mu -\eta )\,%
\mathcal{U}_{1}(\lambda )\,R_{21}(\lambda -\mu ). \label{bYB}
\end{equation}%
Here we consider the  reflection equation associated 
to the 6-vertex trigonometric $R$ matrix 
\begin{equation}
R_{12}(\lambda )=\left( 
\begin{array}{cccc}
\sinh (\lambda +\eta ) & 0 & 0 & 0 \\ 
0 & \sinh \lambda & \sinh \eta & 0 \\ 
0 & \sinh \eta & \sinh \lambda & 0 \\ 
0 & 0 & 0 & \sinh (\lambda +\eta )%
\end{array}%
\right) \in \text{End}(\mathcal{H}_{1}\otimes \mathcal{H}_{2}),
\end{equation}%
where $\mathcal{H}_{a}\simeq \mathbb{C}^{2}$ is a 2-dimensional linear
space. The 6-vertex trigonometric $R$-matrix is a solution of the Yang-Baxter
equation:%
\begin{equation}
R_{12}(\lambda -\mu )R_{13}(\lambda )R_{23}(\mu )=R_{23}(\mu )R_{13}(\lambda
)R_{12}(\lambda -\mu ).
\end{equation}%
The most general scalar solution ($2\times 2$ matrix) of the reflection
equation reads%
\begin{equation}
K(\lambda ;\zeta ,\kappa ,\tau )=\frac{1}{\sinh \zeta }\left( 
\begin{array}{cc}
\sinh (\lambda -\eta /2+\zeta ) & \kappa e^{\tau }\sinh (2\lambda -\eta ) \\ 
\kappa e^{-\tau }\sinh (2\lambda -\eta ) & \sinh (\zeta -\lambda +\eta /2)%
\end{array}%
\right) \in \text{End}(\mathcal{H}_{0}\simeq \mathbb{C}^{2}),  \label{ADMFKK}
\end{equation}%
where $\zeta ,$ $\kappa $ and $\tau $ are arbitrary complex parameters.
Using it and following \cite{Skl88} we can construct two classes of
solutions to the reflection equation (\ref{bYB}) in the 2$^{\mathsf{N}}$%
-dimensional representation space:%
\begin{equation}
\mathcal{H}=\otimes _{n=1}^{\mathsf{N}}\mathcal{H}_{n}.
\end{equation}%
Indeed, starting from%
\begin{equation}
K_{-}(\lambda )=K(\lambda ;\zeta _{-},\kappa _{-},\tau _{-}),\text{ \ \ \ \ }%
K_{+}(\lambda )=K(\lambda +\eta ;\zeta _{+},\kappa _{+},\tau _{+}),
\end{equation}%
where $\zeta _{\pm },\kappa _{\pm },\tau _{\pm }$ are the  boundary
parameters, the following boundary monodromy matrices can be  introduced %
\begin{eqnarray}
\mathcal{U}_{-}(\lambda ) &=&M_{0}(\lambda )K_{-}(\lambda )\widehat{M}%
_{0}(\lambda )=\left( 
\begin{array}{cc}
\mathcal{A}_{-}(\lambda ) & \mathcal{B}_{-}(\lambda ) \\ 
\mathcal{C}_{-}(\lambda ) & \mathcal{D}_{-}(\lambda )%
\end{array}%
\right) \in \text{End}(\mathcal{H}_{0}\otimes \mathcal{H}), \\
\mathcal{U}_{+}^{t_{0}}(\lambda ) &=&M_{0}^{t_{0}}(\lambda
)K_{+}^{t_{0}}(\lambda )\widehat{M}_{0}^{t_{0}}(\lambda )=\left( 
\begin{array}{cc}
\mathcal{A}_{+}(\lambda ) & \mathcal{C}_{+}(\lambda ) \\ 
\mathcal{B}_{+}(\lambda ) & \mathcal{D}_{+}(\lambda )%
\end{array}%
\right) \in \text{End}(\mathcal{H}_{0}\otimes \mathcal{H}).
\end{eqnarray}%
 These matrices $\mathcal{U}_{-}(\lambda )$ and $\mathcal{V}%
_{+}(\lambda )=\mathcal{U}_{+}^{t_{0}}(-\lambda )$ define two classes of
solutions of the reflection equation (\ref{bYB}). Here, we have used the
notations:%
\begin{equation}
M_{0}(\lambda )=R_{0\mathsf{N}}(\lambda -\xi _{\mathsf{N}}-\eta /2)\dots
R_{01}(\lambda -\xi _{1}-\eta /2)=\left( 
\begin{array}{cc}
A(\lambda ) & B(\lambda ) \\ 
C(\lambda ) & D(\lambda )%
\end{array}%
\right)  \label{T}
\end{equation}%
and 
\begin{equation}
\widehat{M}(\lambda )=(-1)^{\mathsf{N}}\,\sigma _{0}^{y}\,M^{t_{0}}(-\lambda
)\,\sigma _{0}^{y},  \label{Mhat}
\end{equation}%
where $M_{0}(\lambda )\in $ End$(\mathcal{H}_{0}\otimes \mathcal{H})$ is the
bulk inhomogeneous monodromy matrix (the $\xi _{j}$ are the arbitrary
inhomogeneity parameters) satisfing the Yang-Baxter relation:%
\begin{equation}
R_{12}(\lambda -\mu )M_{1}(\lambda )M_{2}(\mu )=M_{2}(\mu )M_{1}(\lambda
)R_{12}(\lambda -\mu ).  \label{YB}
\end{equation}%
The main interest of these boundary monodromy matrices is the property
shown by Sklyanin \cite{Skl88} that the following family of transfer
matrices:%
\begin{equation}
\mathcal{T}(\lambda )=\text{tr}_{0}\{K_{+}(\lambda )\,M(\lambda
)\,K_{-}(\lambda )\widehat{M}(\lambda )\}=\text{tr}_{0}\{K_{+}(\lambda )\mathcal{%
U}_{-}(\lambda )\}=\text{tr}_{0}\{K_{-}(\lambda )\mathcal{U}_{+}(\lambda
)\}\in \text{\thinspace End}(\mathcal{H}),  \label{transfer}
\end{equation}%
defines a one parameter family of commuting operators in End$(\mathcal{H})$.
The Hamiltonian of the open XXZ quantum spin 1/2 chain with the most general
integrable boundary terms can be  obtained in the homogeneous limit ($\xi _{m}=0$
for $m=1,\ldots ,\mathsf{N}$) from the following derivative of the transfer
matrix (\ref{transfer}):%
\begin{equation}
H=\frac{2(\sinh \eta )^{1-2\mathsf{N}}}{\text{tr}\{K_{+}(\eta /2)\}\,\text{tr%
}\{K_{-}(\eta /2)\}}\frac{d}{d\lambda }\mathcal{T}(\lambda )_{\,\vrule %
height13ptdepth1pt\>{\lambda =\eta /2}\!}+\text{constant,}  \label{Ht}
\end{equation}%
and its explicit form reads: 
\begin{align}
H& =\sum_{i=1}^{\mathsf{N}-1}(\sigma _{i}^{x}\sigma _{i+1}^{x}+\sigma
_{i}^{y}\sigma _{i+1}^{y}+\cosh \eta \sigma _{i}^{z}\sigma _{i+1}^{z}) 
\notag \\
& +\frac{\sinh \eta }{\sinh \zeta _{-}}\left[ \sigma _{1}^{z}\cosh \zeta
_{-}+2\kappa _{-}(\sigma _{1}^{x}\cosh \tau _{-}+i\sigma _{1}^{y}\sinh \tau
_{-})\right]  \notag \\
& +\frac{\sinh \eta }{\sinh \zeta _{+}}[(\sigma _{\mathsf{N}}^{z}\cosh \zeta
_{+}+2\kappa _{+}(\sigma _{\mathsf{N}}^{x}\cosh \tau _{+}+i\sigma _{\mathsf{N%
}}^{y}\sinh \tau _{+}).  \label{H-XXZ-Non-D}
\end{align}%
Here $\sigma _{i}^{a}$ are local spin $1/2$ operators (Pauli matrices), $%
\Delta =\cosh \eta $ is the anisotropy parameter and the six complex
boundary parameters $\zeta _{\pm }$, $\kappa _{\pm }$ and $\tau _{\pm }$
define the most general integrable magnetic interactions at the boundaries.

\subsection{Some relevant properties}

The following quadratic linear combination of the
generators $\mathcal{A}_{-}(\lambda ),$ $\mathcal{B}_{-}(\lambda ),$ $%
\mathcal{C}_{-}(\lambda )$ and $\mathcal{D}_{-}(\lambda )$ of the reflection
algebra: 
\begin{align}
\frac{\mathrm{det}_{q}\,\mathcal{U}_{-}(\lambda )}{\sinh (2\lambda -2\eta )}& =%
\mathcal{A}_{-}(\epsilon \lambda +\eta /2)\mathcal{A}_{-}(\eta /2-\epsilon
\lambda )+\mathcal{B}_{-}(\epsilon \lambda +\eta /2)\mathcal{C}_{-}(\eta
/2-\epsilon \lambda )  \label{q-detU_1} \\
& =\mathcal{D}_{-}(\epsilon \lambda +\eta /2)\mathcal{D}_{-}(\eta
/2-\epsilon \lambda )+\mathcal{C}_{-}(\epsilon \lambda +\eta /2)\mathcal{B}%
_{-}(\eta /2-\epsilon \lambda ),  \label{q-detU_2}
\end{align}%
where $\epsilon =\pm 1$, is the  \textit{quantum determinant }. It was
shown by Sklyanin that it is a central element of the reflection algebra%
\begin{equation}
\lbrack \mathrm{det}_{q}\,\mathcal{U}_{-}(\lambda ),\mathcal{U}_{-}(\mu )]=0.
\end{equation}%
The quantum determinant plays a fundamental role in the characterization of
the transfer matrix spectrum and it admits the following explicit
expressions:%
\begin{eqnarray}
\mathrm{det}_{q}\,\mathcal{U}_{-}(\lambda ) &=&\mathrm{det}_{q}K_{-}(\lambda )%
\mathrm{det}_{q}M_{0}(\lambda )\mathrm{det}_{q}M_{0}(-\lambda )
\label{q-detU_-exp} \\
&=&\sinh (2\lambda -2\eta )\mathsf{A}_{-}(\lambda +\eta /2)\mathsf{A}%
_{-}(-\lambda +\eta /2),
\end{eqnarray}%
where: 
\begin{equation}
\mathrm{det}_{q}M(\lambda )=a(\lambda +\eta /2)d(\lambda -\eta /2),
\label{bulk-q-det}
\end{equation}%
is the bulk quantum determinant and%
\begin{equation}
\mathrm{det}_{q}K_{\pm }(\lambda )=\mp \sinh (2\lambda \pm 2\eta )g_{\pm
}(\lambda +\eta /2)g_{\pm }(-\lambda +\eta /2).
\end{equation}%
Here, we  used the following notations:%
\begin{equation}
\mathsf{A}_{-}(\lambda )=g_{-}(\lambda )a(\lambda )d(-\lambda ),\text{ \ }%
d(\lambda )=a(\lambda -\eta ),\text{ \ \ }a(\lambda )=\prod_{n=1}^{\mathsf{N}%
}\sinh (\lambda -\xi _{n}+\eta /2),  \label{eigenA}
\end{equation}%
and%
\begin{equation}
g_{\pm }(\lambda )=\frac{\sinh (\lambda +\alpha _{\pm }-\eta /2)\cosh
(\lambda \mp \beta _{\pm }-\eta /2)}{\sinh \alpha _{\pm }\cosh \beta _{\pm }}%
,  \label{g_PM}
\end{equation}%
where $\alpha _{\pm }$\ and $\beta _{\pm }$ are defined in terms of the
boundary parameters by:%
\begin{equation}
\sinh \alpha _{\pm }\cosh \beta _{\pm }=\frac{\sinh \zeta _{\pm }}{2\kappa
_{\pm }},\text{ \ \ \ \ \ }\cosh \alpha _{\pm }\sinh \beta _{\pm }=\frac{%
\cosh \zeta _{\pm }}{2\kappa _{\pm }}.  \label{alfa-beta}
\end{equation}

\begin{proposition}[Prop. 2.3 of \protect\cite{Nic12b}]
\label{normality}The transfer matrix $\mathcal{T}(\lambda )$ is an even function of the
spectral parameter $\lambda $:%
\begin{equation}
\mathcal{T}(-\lambda )=\mathcal{T}(\lambda ),  \label{even-transfer}
\end{equation}%
and it is central for the following special values of the spectral parameter: 
\begin{eqnarray}
\lim_{\lambda \rightarrow \pm \infty }e^{\mp 2\lambda (\mathsf{N}+2)}%
\mathcal{T}(\lambda ) &=&2^{-(2\mathsf{N}+1)}\frac{\kappa _{+}\kappa
_{-}\cosh (\tau _{+}-\tau _{-})}{\sinh \zeta _{+}\sinh \zeta _{-}},
\label{Central-asymp} \\
\mathcal{T}(\pm \eta /2) &=&(-1)^{\mathsf{N}}2\cosh \eta \mathrm{det}_{q}M(0),
\label{Central-1} \\
\mathcal{T}(\pm (\eta /2-i\pi /2)) &=&-2\cosh \eta \coth \zeta _{-}\coth
\zeta _{+}\mathrm{det}_{q}M(i\pi /2).  \label{Central-2}
\end{eqnarray}%
Moreover, the monodromy matrix \thinspace $\mathcal{U}_{\pm }(\lambda )$
satisfy the following transformation properties under Hermitian conjugation:%
\begin{itemize}
\item Under the condition $\eta \in i\mathbb{R}$ (massless regime), it
holds: 
\begin{equation}
\mathcal{U}_{\pm }(\lambda )^{\dagger }=\left[ \mathcal{U}_{\pm }(-\lambda
^{\ast })\right] ^{t_{0}},  \label{ml-Hermitian_U}
\end{equation}%
 for $\{i\tau _{\pm },i\kappa _{\pm },i\zeta _{\pm },\xi
_{1},...,\xi _{\mathsf{N}}\}\in \mathbb{R}^{\mathsf{N}+3}.$
\item Under the condition $\eta \in \mathbb{R}$ (massive regime), it
holds: 
\begin{equation}
\mathcal{U}_{\pm }(\lambda )^{\dagger }=\left[ \mathcal{U}_{\pm }(\lambda
^{\ast })\right] ^{t_{0}},  \label{m-Hermitian_U}
\end{equation}%
for $\{\tau _{\pm },\kappa _{\pm },\zeta _{\pm },i\xi
_{1},...,i\xi _{\mathsf{N}}\}\in \mathbb{R}^{\mathsf{N}+3}.$
\end{itemize}
So under the same conditions on the parameters of the representation it
holds: 
\begin{equation}
\mathcal{T}(\lambda )^{\dagger }=\mathcal{T}(\lambda ^{\ast }),
\label{I-Hermitian_T}
\end{equation}%
i.e. $\mathcal{T}(\lambda )$ defines a one-parameter\ family of normal
operators which are self-adjoint both for $\lambda $ real and purely
imaginary.
\end{proposition}

\subsection{SOV representations for $\mathcal{T}(\protect\lambda )$-spectral
problem}

Let us recall here the characterization obtained in \cite{Nic12b,Fald-KN13}
by SOV method of the spectrum of the transfer matrix $\mathcal{T}(\lambda )$. First we introduce the  following notations:%
\begin{equation}
X_{k,m}^{(i,r)}(\tau _{\pm },\alpha _{\pm },\beta _{\pm })\equiv \left(
-1\right) ^{i}\left( 1-r\right) \eta +\tau _{-}-\tau_{+}+(-1)^{k}(\alpha
_{-}+\beta _{-})-(-1)^{m}(\alpha _{+}-\beta _{+})+i\pi (k+m),
\label{SOV-cond-}
\end{equation}%
and by using these linear combinations of the boundary parameters we
introduce the set $N_{SOV}\subset\mathbb{C}^6$ of boundary parameters for which the separation of variables cannot be applied directly. More precisely   
$$(\tau _{+},\alpha _{+},\beta _{+},\tau _{-},\alpha_{-},\beta _{-})\in N_{SOV},$$ 
if $\exists (k,h,m,n)\in \left\{ 0,1\right\} $ such that %
\begin{equation}
X_{k,m}^{(0,\mathsf{N})}(\tau _{\pm },\alpha _{\pm },\beta _{\pm })=0 \quad\text{and}\quad
X_{h,n}^{(1,\mathsf{N})}(\tau _{\pm },\alpha _{\pm },\beta _{\pm })=0.\end{equation}%
All the results in the following will be obtained for the generic values of the boundary parameters, not belonging to this set. The SOV method applicability can be further extended applying the discrete symmetries discussed in the Section \ref{sect-descretesym}.

 Following \cite{Fald-KN13} we define the functions:%
\begin{eqnarray}
g_{a}(\lambda ) &=&\frac{\cosh ^{2}2\lambda -\cosh ^{2}\eta }{\cosh
^{2}2\zeta _{a}^{(0)}-\cosh ^{2}\eta }\,\prod_{\substack{ b=1 \\ b\neq a}}^{%
\mathsf{N}}\frac{\cosh 2\lambda -\cosh 2\zeta _{b}^{(0)}}{\cosh 2\zeta
_{a}^{(0)}-\cosh 2\zeta _{b}^{(0)}}\quad \text{ \ for }a\in \{1,...,\mathsf{N%
}\}, \\
\mathbf{A}(\lambda ) &=&(-1)^{\mathsf{N}}\frac{\sinh (2\lambda +\eta )}{%
\sinh 2\lambda }g_{+}(\lambda )g_{-}(\lambda )a(\lambda )d(-\lambda ),
\end{eqnarray}%
and%
\begin{align}
f(\lambda )=& \frac{\cosh 2\lambda +\cosh \eta }{2\cosh \eta }\prod_{b=1}^{%
\mathsf{N}}\frac{\cosh 2\lambda -\cosh 2\zeta _{b}^{(0)}}{\cosh \eta -\cosh
2\zeta _{b}^{(0)}}\mathbf{A}(\eta /2)  \notag \\
& -(-1)^{\mathsf{N}}\frac{\cosh 2\lambda -\cosh \eta }{2\cosh \eta }%
\prod_{b=1}^{\mathsf{N}}\frac{\cosh 2\lambda -\cosh 2\zeta _{b}^{(0)}}{\cosh
\eta +\cosh 2\zeta _{b}^{(0)}}\mathbf{A}(\eta /2+i\pi /2)  \notag \\
& +2^{(1-\mathsf{N})}\frac{\kappa _{+}\kappa _{-}\cosh (\tau _{+}-\tau _{-})%
}{\sinh \zeta _{+}\sinh \zeta _{-}}(\cosh ^{2}2\lambda -\cosh ^{2}\eta
)\prod_{b=1}^{\mathsf{N}}(\cosh 2\lambda -\cosh 2\zeta _{b}^{(0)}),
\label{f-function}
\end{align}%
where%
\begin{equation}
\zeta _{n}^{(h_{n})}=\xi _{n}+(h_{n}-\frac{1}{2})\eta \quad \forall n\in
\{1,...,\mathsf{N}\},\text{ }h_{n}\in \{0,1\}\text{.}
\end{equation}%
We can now recall the main result on the characterization of the set $\Sigma
_{\mathcal{T}}$ formed by all the eigenvalue functions of the transfer
matrix $\mathcal{T}(\lambda )$.

\begin{theorem}[Theorem 5.3 and Corollary 5.1 of \protect\cite{Fald-KN13}]
\label{C:T-eigenstates-}Let $(\tau _{+},\alpha _{+},\beta _{+},\tau
_{-},\alpha _{-},\beta _{-})\in \mathbb{C}^{6}\backslash N_{SOV}$ and let
the inhomogeneities $\{\xi _{1},...,\xi _{\mathsf{N}}\}\in \mathbb{C}$ $^{%
\mathsf{N}}$ be generic:%
\begin{equation}
\xi _{a}\neq \pm\xi _{b}+r\eta\text{\ \ mod\,}2\pi \text{ \ }\forall a\neq b\in \{1,...,\mathsf{N}%
\}\,\,\text{and\thinspace \thinspace }r\in \{-1,0,1\},  \label{xi-conditions}
\end{equation}%
then $\mathcal{T}(\lambda )$ has simple spectrum and the set of its eigenvalues $\Sigma _{\mathcal{T}}$
is characterized by:%
\begin{equation}
\Sigma _{\mathcal{T}}=\left\{ \tau (\lambda ):\tau (\lambda )=f(\lambda
)+\sum_{a=1}^{\mathsf{N}}g_{a}(\lambda )x_{a},\text{ \ \ }\forall
\{x_{1},...,x_{\mathsf{N}}\}\in \Sigma _{T}\right\} ,
\label{Interpolation-Form-T}
\end{equation}%
where $\Sigma _{T}$ is the set of solutions to the following inhomogeneous
system of $\mathsf{N}$ quadratic equations:%
\begin{equation}
x_{n}\sum_{a=1}^{\mathsf{N}}g_{a}(\zeta _{n}^{(1)})x_{a}+x_{n}f(\zeta
_{n}^{(1)})=q_{n},\text{ \ \ \ }q_{n}=\frac{\mathrm{det}_{q}K_{+}(\xi _{n})%
\mathrm{det}_{q}\,\mathcal{U}_{-}(\xi _{n})}{\sinh (\eta +2\xi _{n})\sinh (\eta
-2\xi _{n})},\text{ \ \ }\forall n\in \{1,...,\mathsf{N}\},
\label{Quadratic System}
\end{equation}%
in $\mathsf{N}$ unknowns $\{x_{1},...,x_{\mathsf{N}}\}$.
\end{theorem}

\section{Inhomogeneous Baxter equation 
}

Here we show that the SOV characterization of the spectrum admits an
equivalent  formulation in terms of a second order functional
difference equation of Baxter type:%
\begin{equation}
\tau (\lambda )Q(\lambda )=\mathbf{A}(\lambda )Q(\lambda -\eta )+\mathbf{A}%
(-\lambda )Q(\lambda +\eta )+F(\lambda ),  \label{Inhom-Baxter-Eq}
\end{equation}%
which contains a non-zero inhomogeneous term $F(\lambda )$ non-zero for
generic integrable boundary conditions and the $Q$-functions are
{\it trigonometric polynomials}.
In this paper we will call $f(\lambda)$ a trigonometric polynomial of degree $\mathsf{M}$ if
$e^{\mathsf{M}\lambda}\,f(\lambda)$ is a polynomial  of $e^{2\lambda}$ of degree  $\mathsf{M}$. Most trigonometric polynomials we will consider in the following sections will be even functions of $\lambda$ and will satisfy an additional condition $f(\la+i\pi)=f(\la)$. It is easy to see in this situation that such functions can be written as polynomials of $\cosh 2\lambda$.

\subsection{Main functions in the functional equation}

Let $Q(\lambda )$ be an even trigonometric polynomial of degree $2\mathsf{N}$. 
It can be written in  the following form:%
\begin{align}
Q(\lambda )& =\sum_{a=1}^{\mathsf{N}}\prod_{\substack{ b=1  \\ b\neq a}}^{%
\mathsf{N}}\frac{\cosh 2\lambda -\cosh 2\zeta _{b}^{(0)}}{\cosh 2\zeta
_{a}^{(0)}-\cosh 2\zeta _{b}^{(0)}}Q(\zeta _{a}^{(0)})+2^{\mathsf{N}%
}\prod_{a=1}^{\mathsf{N}}\left( \cosh 2\lambda -\cosh 2\zeta
_{a}^{(0)}\right)  \label{Q-form1} \\
& =2^{\mathsf{N}}\prod_{a=1}^{\mathsf{N}}\left( \cosh 2\lambda -\cosh
2\lambda _{a}\right) ,  \label{Q-form2}
\end{align}%
where from now on the $Q(\zeta _{a}^{(0)})$ are arbitrary complex numbers or
similarly the $\lambda _{a}$ are arbitrary complex numbers. Then, introducing
the function:%
\begin{equation}
Z_{Q}(\lambda )=\mathbf{A}(\lambda )Q(\lambda -\eta )+\mathbf{A}(-\lambda
)Q(\lambda +\eta )
\end{equation}%
we can prove the following Lemma

\begin{lemma}
Let $Q(\lambda )$ be any function of the form $\left( \ref{Q-form2}\right) $
then the associated function $Z_{Q}(\lambda )$ is an even trigonometric
polynomial of degree $4\mathsf{N}+4$ of the following form:%
\begin{equation}
Z_{Q}(\lambda )=\sum_{a=0}^{2(\mathsf{N}+1)}z_{a}\cosh ^{a}2\lambda ,\text{
with }z_{2(\mathsf{N}+1)}=\frac{2\kappa _{+}\kappa _{-}\cosh (\alpha
_{+}+\alpha _{-}-\beta _{+}+\beta _{-}-(\mathsf{N}+1)\eta )}{\sinh \zeta
_{+}\sinh \zeta _{-}}.
\end{equation}
\end{lemma}

\begin{proof}[Proof]
The fact that the function $Z_{Q}(\lambda )$ is even in $\lambda $ is a
trivial consequence of the fact that $Q(\lambda )$ is even; in fact, it holds:%
\begin{eqnarray}
Z_{Q}(-\lambda ) &=&\mathbf{A}(-\lambda )Q(-\lambda -\eta )+\mathbf{A}%
(\lambda )Q(-\lambda +\eta )  \notag \\
&=&\mathbf{A}(-\lambda )Q(\lambda +\eta )+\mathbf{A}(\lambda )Q(\lambda
-\eta )=Z_{Q}(\lambda ).
\end{eqnarray}%
The fact that $Z_{Q}(\lambda )$ is indeed a trigonometric polynomial follows
from its definition once we observe that $\lambda=0$ is not a singular point and the following identity holds:%
\begin{equation}
\lim_{\lambda \rightarrow 0}Z_{Q}(\lambda )=2g_{+}(0)g_{-}(0)a(0)a(-\eta
)Q(0)\cosh \eta .
\end{equation}%
Now the functional form of $Z_{Q}(\lambda )$ is a consequence of the
following identities:%
\begin{equation}
Z_{Q}(\lambda +i\pi )=Z_{Q}(\lambda ),\text{ \ }\lim_{\lambda \rightarrow
\pm \infty }\frac{Z_{Q}(\lambda )}{e^{\pm 4(\mathsf{N}+1)\lambda }}=\frac{%
\kappa _{+}\kappa _{-}\cosh (\alpha _{+}+\alpha _{-}-\beta _{+}+\beta _{-}-(%
\mathsf{N}+1)\eta )}{2^{(2\mathsf{N}+1)}\sinh \zeta _{+}\sinh \zeta _{-}},
\end{equation}%
where the second identity follows from:%
\begin{align}
\lim_{\lambda \rightarrow \pm \infty }e^{\mp (2\mathsf{N}+4)\lambda }\mathbf{%
A}(\lambda )& =2^{-2(\mathsf{N}+1)}\frac{\kappa _{+}\kappa _{-}\exp \pm
(\alpha _{+}+\alpha _{-}-\beta _{+}+\beta _{-}+(\mathsf{N}-1)\eta )}{\sinh
\zeta _{+}\sinh \zeta _{-}}, \\
\lim_{\lambda \rightarrow \pm \infty }e^{\mp 2\mathsf{N}\lambda }Q(\lambda
)& =1.
\end{align}
\end{proof}

\subsection{On the need of an inhomogeneous term in the functional equation}

Here, we would like to point out that it is simple to define the boundary
conditions for which one can prove that the homogeneous version of the
Baxter equation $\left( \ref{Inhom-Baxter-Eq}\right) $ does not admit
trigonometric polynomial solutions for $\tau (\lambda )\in \Sigma _{\mathcal{%
T}}$.

\begin{lemma}
\label{impossible_hom}
Assume that the boundary parameters satisfy the following conditions:%
\begin{equation}
\kappa _{+}\neq 0,\kappa _{-}\neq 0,\text{ \ }Y^{(i,r)}(\tau _{\pm },\alpha
_{\pm },\beta _{\pm })\neq 0\text{ \ }\forall i\in \left\{ 0,1\right\} ,r\in 
\mathbb{Z}  \label{Inhomogeneous-boundary conditions}
\end{equation}%
where we have defined:%
\begin{equation}
Y^{(i,r)}(\tau _{\pm },\alpha _{\pm },\beta _{\pm })\equiv \tau _{-}-\tau
_{+}+\left( -1\right) ^{i}\left[ \left( \mathsf{N}-1-r\right) \eta +(\alpha
_{-}+\alpha _{+}+\beta _{-}-\beta _{+})\right] ,
\end{equation}%
then for any $\tau (\lambda )\in \Sigma _{\mathcal{T}}$ \ the homogeneous
Baxter equation:%
\begin{equation}
\tau (\lambda )Q(\lambda )=\mathbf{A}(\lambda )Q(\lambda -\eta )+\mathbf{A}%
(-\lambda )Q(\lambda +\eta ),
\end{equation}%
does not admit any (non identically zero) $Q(\lambda )$ of Laurent
polynomial form in $e^{\lambda }$.
\end{lemma}

\begin{proof}[Proof]
If we consider the following function:%
\begin{equation}
Q(\lambda )=\sum_{a=-s}^{r}y_{a}e^{a\lambda },\text{ \ with }r,s\in \mathbb{N%
}
\end{equation}%
we can clearly always chose the coefficients $y_{a}$ such that the r.h.s. of
the homogeneous Baxter equation has no poles as required. However, it is
enough to consider now the asymptotics:%
\begin{align}
\lim_{\lambda \rightarrow +\infty }\frac{\left[ \mathbf{A}(\lambda
)Q(\lambda -\eta )+\mathbf{A}(-\lambda )Q(\lambda +\eta )\right] }{e^{(2%
\mathsf{N}+4+r)\lambda }}& =\frac{y_{r}\kappa _{+}\kappa _{-}\cosh (\alpha
_{+}+\alpha _{-}-\beta _{+}+\beta _{-}+(\mathsf{N}-1-r)\eta )}{2^{2(\mathsf{N%
}+1)}\sinh \zeta _{+}\sinh \zeta _{-}} \\
\lim_{\lambda \rightarrow +\infty }e^{-(2\mathsf{N}+4+r)\lambda }\tau
(\lambda )Q(\lambda )& =\frac{y_{r}\kappa _{+}\kappa _{-}\cosh (\tau
_{+}-\tau _{-})}{2^{2(\mathsf{N}+1)}\sinh \zeta _{+}\sinh \zeta _{-}}
\end{align}%
and use the conditions $\left( \ref{Inhomogeneous-boundary conditions}%
\right) $ to observe that for any $r\in \mathbb{Z}$ the asymptotic of the
homogeneous Baxter equation cannot be satisfied which implies the validity
of the lemma.
\end{proof}

\subsection{SOV spectrum in terms of the inhomogeneous Baxter equation}

We introduce now the following function of the boundary parameters:%
\begin{equation}
F_{0}=\frac{2\kappa _{+}\kappa _{-}\left( \cosh (\tau _{+}-\tau _{-})-\cosh
(\alpha _{+}+\alpha _{-}-\beta _{+}+\beta _{-}-(\mathsf{N}+1)\eta )\right) }{%
\sinh \zeta _{+}\sinh \zeta _{-}},
\end{equation}%
and then the function:%
\begin{eqnarray}
F(\lambda ) &=&2^{\mathsf{N}}\,F_{0}\,(\cosh ^{2}2\lambda -\cosh ^{2}\eta
)a(\lambda )a(-\lambda )d(-\lambda )d(\lambda ) \\
&=&F_{0}\, (\cosh ^{2}2\lambda -\cosh ^{2}\eta )\prod_{b=1}^{\mathsf{N}%
}\prod_{i=0}^{1}(\cosh 2\lambda -\cosh 2\zeta _{b}^{(i)}).
\end{eqnarray}%
We introduce also  the set of functions $\Sigma _{\mathcal{Q}}$ such that $Q(\lambda)\in\Sigma _{\mathcal{Q}}$ if  it has a form $\left( %
\ref{Q-form2}\right) $ and
$$\tau(\lambda)=\frac{Z_{Q}(\lambda )+F(\lambda )}{%
Q(\lambda )}$$
is a trigonometric polynomial.
We are now ready to prove the main theorem of this article:

\begin{theorem}
\label{T-eigenvalue-F-eq}Let the inhomogeneities $\{\xi _{1},...,\xi _{%
\mathsf{N}}\}\in \mathbb{C}$ $^{\mathsf{N}}$ be generic \rf{xi-conditions}
and let the boundary parameters $(\tau _{+},\alpha _{+},\beta _{+},\tau _{-},\alpha _{-},\beta
_{-})\in \mathbb{C}^{6}\backslash N_{SOV}$ satisfy the following conditions:%
\begin{equation}
\kappa _{+}\neq 0,\kappa _{-}\neq 0,\text{ \ }Y^{(i,2r)}(\tau _{\pm },\alpha
_{\pm },\beta _{\pm })\neq 0\text{ \ }\forall i\in \left\{ 0,1\right\} ,r\in
\left\{ 0,...,\mathsf{N}-1\right\} ,  \label{Inhom-cond-BaxEq}
\end{equation}%
then 
 $\tau (\lambda )\in \Sigma _{\mathcal{T}}$ if and only if $\exists
!Q(\lambda )\in \Sigma _{\mathcal{Q}}$ such that  
\begin{equation}
\tau (\lambda
)Q(\lambda )=Z_{Q}(\lambda )+F(\lambda ).
\end{equation}
\end{theorem}

\begin{proof}[Proof]
First we prove that if $\tau (\lambda )\in \Sigma _{\mathcal{T}}$ \ then
there is a trigonometric polynomial $Q(\lambda )\in \Sigma _{\mathcal{Q}}$  satisfying the
inhomogeneous functional Baxter equation: 
\begin{equation}
\tau (\lambda )Q(\lambda )=\mathbf{A}(\lambda )Q(\lambda -\eta )+\mathbf{A}%
(-\lambda )Q(\lambda +\eta )+F(\lambda ).
\end{equation}%
To prove it we will show that there is the unique set of values $Q(\zeta _{b}^{(0)})$ such that $Q(\lambda)$ of the form (\ref{Q-form1}) satisfies this equation.

It is straightforward to verify that if $\tau (\lambda )\in \Sigma _{\mathcal{T}}$ \
and $Q(\lambda )$ has the form $\left( \ref{Q-form2}\right) $ then the left
and right hand sides of the above equation are both even trigonometric polynomials of $\lambda 
$ and both can be written (using the asymptotic behavior) in  the form:%
\begin{equation}
\frac{2\kappa _{+}\kappa _{-}\cosh (\tau _{+}-\tau _{-})\prod_{b=1}^{2%
\mathsf{N}+2}(\cosh 2\lambda -\cosh 2y_{b}^{\left( lhs/rhs\right) })}{\sinh
\zeta _{+}\sinh \zeta _{-}}.
\end{equation}%
Then to prove that we can introduce a $Q(\lambda )$ of the form $\left( \ref%
{Q-form2}\right) $ which satisfies the inhomogeneous Baxter  equation $\left( %
\ref{Inhom-Baxter-Eq}\right) $ with $\tau (\lambda )\in \Sigma _{\mathcal{T}%
} $, we have only to prove that $\left( \ref{Inhom-Baxter-Eq}\right) $ is
satisfied in $4\mathsf{N}+4$ different values of $\lambda $. As the {\it r.h.s} and {\it l.h.s} of $\left( \ref{Inhom-Baxter-Eq}\right) $ are even functions we need to check this identity only for $2N+2$ non-zero points  $\mu_j$ such that $\mu_j\neq \pm \mu_k$. It is a simple
exercise verify that the equation $\left( \ref{Inhom-Baxter-Eq}\right) $ is
satisfied automatically for any $Q(\lambda )$ of the form $\left( \ref%
{Q-form2}\right) $ in the following two points, $ \eta /2$ and $ \eta
/2+i\pi /2$:%
\begin{equation}
\tau (\eta /2)Q(\eta /2)=
\mathbf{A}(\eta /2)Q(\eta /2-\eta )=\mathbf{A}(\eta /2)Q(\eta /2)
,  \label{System-A}
\end{equation}%
and:%
\begin{equation}
\tau (\eta /2+i\pi /2)Q(\eta /2+i\pi /2)=
\mathbf{A}(\eta /2+i\pi /2)Q(i\pi /2-\eta /2) \\ 
=\mathbf{A}(\eta /2+i\pi /2)Q(\eta /2+i\pi /2)
 .  \label{System-B}
\end{equation}%
Indeed, these equations reduce to:%
\begin{equation}
\tau (\eta /2)=\mathbf{A}(\eta /2),\text{ \ \ \ }\tau (\eta /2+i\pi /2)=%
\mathbf{A}(\eta /2+i\pi /2)
\end{equation}%
and so they are satisfied by definition for any $\tau (\lambda )\in \Sigma _{\mathcal{T}}$. 
  Then we check the explicit form of the equation $\left( \ref%
{Inhom-Baxter-Eq}\right) $ in the $2\mathsf{N}$ points $ \zeta _{b}^{(0)}$
and $\zeta _{b}^{(1)}$:%
\begin{equation}
\tau (\zeta _{b}^{(0)})Q(\zeta _{b}^{(0)})=
\mathbf{A}(-\zeta _{b}^{(0)})Q(\zeta _{b}^{(0)}+\eta )=\mathbf{A}(-\zeta
_{b}^{(0)})Q(\zeta _{b}^{(1)})
,\end{equation}%
and:%
\begin{equation}
\tau (\zeta _{b}^{(1)})Q(\zeta _{b}^{(1)})=
\mathbf{A}(\zeta _{b}^{(1)})Q(\zeta _{b}^{(1)}-\eta )=\mathbf{A}(\zeta
_{b}^{(1)})Q(\zeta _{b}^{(0)})
.
\end{equation}%
They are equivalent to the following system of equations:%
\begin{align}
\frac{\mathbf{A}(\zeta _{b}^{(1)})}{\tau (\zeta _{b}^{(1)})}& =\frac{\tau
(\zeta _{b}^{(0)})}{\mathbf{A}(-\zeta _{b}^{(0)})}\text{ \ \ \ \ \ }\forall
b\in \{1,...,\mathsf{N}\}  \label{System1} \\
\frac{Q(\zeta _{b}^{(0)})\tau (\zeta _{b}^{(0)})}{\mathbf{A}(-\zeta
_{b}^{(0)})}& =\sum_{a=1}^{\mathsf{N}}\prod_{\substack{ c=1  \\ c\neq a}}^{%
\mathsf{N}}\frac{\cosh 2\zeta _{b}^{(1)}-\cosh 2\zeta _{c}^{(0)}}{\cosh
2\zeta _{a}^{(0)}-\cosh 2\zeta _{c}^{(0)}}Q(\zeta _{a}^{(0)})+2^{\mathsf{N}%
}\prod_{a=1}^{\mathsf{N}}\left( \cosh 2\zeta _{b}^{(1)}-\cosh 2\zeta
_{a}^{(0)}\right)  \label{System2}
\end{align}%
Now using the following  quantum determinant identity
\begin{equation}
\frac{\det_{q}K_{+}(\lambda-\eta/2)\det_{q}
\mathcal{U}_{-}(\lambda -\eta /2)}{\sinh (2\lambda +\eta )\sinh (2\lambda -\eta )}=\mathbf{A}(\lambda )\mathbf{A}(-\lambda +\eta ).\label{Tot-q-det-tt}
\end{equation}
it is easy to see that the system of equations $\left( \ref{System1}\right) $ is certainly
satisfied as $\tau (\lambda )\in \Sigma _{\mathcal{T}}$, once we recall the
SOV characterization (\ref{Interpolation-Form-T}) of $\Sigma _{\mathcal{T}}$. Indeed there is a set 
$\{x_1,\dots,x_n\}$ satisfying the  equations (\ref{Quadratic System}) and  $\tau(\zeta _{b}^{(0)})=x_b$.

So we are left with $\left( %
\ref{System2}\right) $ a linear system of $\mathsf{N}$ inhomogeneous
equations with $\mathsf{N}$ unknowns $Q(\zeta _{a}^{(0)})$. Here, we prove
that the matrix of this linear system%
\begin{equation}
c_{a b}\equiv \prod_{\substack{ c=1  \\ c\neq a}}^{\mathsf{N}}\frac{\cosh
2\zeta _{b}^{(1)}-\cosh 2\zeta _{c}^{(0)}}{\cosh 2\zeta _{a}^{(0)}-\cosh
2\zeta _{c}^{(0)}}-\delta _{a b}\frac{\tau (\zeta _{b}^{(0)})}{\mathbf{A}%
(-\zeta _{b}^{(0)})}\text{ \ \ \ \ \ }\forall a,b\in \{1,...,\mathsf{N}\}
\end{equation}%
 has nonzero determinant for the given $\tau (\lambda
)\in \Sigma _{\mathcal{T}}$. Indeed, let us suppose
that for some $\tau (\lambda )\in \Sigma _{\mathcal{T}}$:%
\begin{equation}
\mathrm{det}_{\mathsf{N}}\left[ c_{a b}\right] =0.  \label{det-coeff}
\end{equation}%
Then there is at least one nontrivial solution $\{Q(\zeta
_{1}^{(0)}),...,Q(\zeta _{\mathsf{N}}^{(0)})\}\neq \{0,...,0\}$ to the
homogeneous system of equations:%
\begin{equation}
\frac{Q(\zeta _{b}^{(0)})\tau (\zeta _{b}^{(0)})}{\mathbf{A}(-\zeta
_{b}^{(0)})}=\sum_{a=1}^{\mathsf{N}}\prod_{\substack{ c=1  \\ c\neq a}}^{%
\mathsf{N}}\frac{\cosh 2\zeta _{b}^{(1)}-\cosh 2\zeta _{c}^{(0)}}{\cosh
2\zeta _{a}^{(0)}-\cosh 2\zeta _{c}^{(0)}}Q(\zeta _{a}^{(0)})
\label{System2-homo}
\end{equation}%
and hence we can define:%
\begin{equation}
Q_{\mathsf{M}}(\lambda )=\sum_{a=1}^{\mathsf{N}}\prod_{\substack{ b=1  \\ %
b\neq a}}^{\mathsf{N}}\frac{\cosh 2\lambda -\cosh 2\zeta _{b}^{(0)}}{\cosh
2\zeta _{a}^{(0)}-\cosh 2\zeta _{b}^{(0)}}Q(\zeta _{a}^{(0)})=\lambda _{%
\mathsf{M}+1}^{(\mathsf{M})}\prod_{b=1}^{\mathsf{M}}\left( \cosh 2\lambda
-\cosh 2\lambda _{b}^{(\mathsf{M})}\right) .
\end{equation}%
It is an even trigonometric polynomial of degree $2\mathsf{M}$ such that  $0\leq 
\mathsf{M}\leq \mathsf{N}-1$ fixed by the solution $\{Q(\zeta
_{1}^{(0)}),...,Q(\zeta _{\mathsf{N}}^{(0)})\}$. Now  using the  $%
Q_{\mathsf{M}}(\lambda )$ and  $\tau (\lambda )\in \Sigma _{\mathcal{T}}$
we can define  two functions:%
\begin{equation}
W_{1}(\lambda )=Q_{\mathsf{M}}(\lambda )\tau (\lambda )\text{ \, and \, }%
W_{2}(\lambda )=\mathbf{A}(\lambda )Q_{\mathsf{M}}(\lambda -\eta )+\mathbf{A}%
(-\lambda )Q_{\mathsf{M}}(\lambda +\eta )
\end{equation}%
which are both even trigonometric polynomials of degree $2\mathsf{M}+2\mathsf{N}+4$.
Then it is straightforward to observe that the systems of equations $\left( \ref%
{System1}\right) $ and $\left( \ref{System2-homo}\right) $ plus the
conditions $\left( \ref{System-A}\right) $ and $\left( \ref{System-B}\right) 
$, which are also satisfied with the function $Q_{\mathsf{M}}(\lambda )$,
imply that $W_{1}(\lambda )$ and $W_{2}(\lambda )$ coincide in $4\mathsf{N}%
+4 $ different values of $\lambda $ ($\pm \eta /2$, $\pm (\eta /2+i\pi /2)$, 
$\pm \zeta _{b}^{(0)}$ and $\pm \zeta _{b}^{(1)}$). It means that $%
W_{1}(\lambda )\equiv W_{2}(\lambda )$, as these  are two polynomials of
maximal degree $4\mathsf{N}+2$. So, we have shown that from the assumption $%
\exists \tau (\lambda )\in \Sigma _{\mathcal{T}}$ such that $\left( \ref%
{det-coeff}\right) $ holds it follows that $\tau (\lambda )$ and $Q_{\mathsf{%
M}}(\lambda )$ have to satisfy the following homogeneous Baxter equations: 
\begin{equation}
\tau (\lambda )Q_{\mathsf{M}}(\lambda )=\mathbf{A}(\lambda )Q_{\mathsf{M}%
}(\lambda -\eta )+\mathbf{A}(-\lambda )Q_{\mathsf{M}}(\lambda +\eta ).
\label{Baxter-eq-homo}
\end{equation}%
Now we can apply the Lemma \ref{impossible_hom} which implies that $Q_{\mathsf{M}}(\lambda )=0$ for any $\lambda$, which contradicts the hypothesis of the existence of a nontrivial solution to the homogeneous system \rf{System2-homo}. Hence, we have proven that $\mathrm{det}_{\mathsf{N}}\left[ c_{a b}\right] \neq 0.$
Therefore there is a unique solution  $\{Q(\zeta
_{1}^{(0)}),...,Q(\zeta _{\mathsf{N}}^{(0)})\}$ of the inhomogeneous system 
$\left( \ref{System2}\right) $ which defines one and only one $Q(\lambda )$ of the
form $\left( \ref{Q-form1}\right) $ satisfying the functional inhomogeneous
Baxter's equation $\left( \ref{Inhom-Baxter-Eq}\right) $.


We prove now that if $Q(\lambda )\in \Sigma _{\mathcal{Q}}$ then $%
\tau (\lambda )=\left( Z_{Q}(\lambda )+F(\lambda )\right) /Q(\lambda )\in
\Sigma _{\mathcal{T}}$. By  definition of the functions $Z_{Q}(\lambda ),$
$F(\lambda )$ and $Q(\lambda )$ the function $\tau (\lambda )$ has the
desired form:%
\begin{equation}
\tau (\lambda )=f(\lambda )+\sum_{a=1}^{\mathsf{N}}g_{a}(\lambda )\tau
(\zeta _{a}^{(0)}).
\end{equation}%
To prove now that $\tau (\lambda )\in \Sigma _{\mathcal{T}}$ we have to
write the inhomogeneous Baxter equation $\left( \ref{Inhom-Baxter-Eq}%
\right) $ in the $2\mathsf{N}$ points $ \zeta _{b}^{(0)}$ and $ \zeta
_{b}^{(1)}$. Indeed, we have already proved that this reproduce the systems $%
\left( \ref{System1}\right) $ and $\left( \ref{System2}\right) $ and it is
simple to observe that the system of equations $\left( \ref{System1}\right) $
just coincides with the inhomogeneous system of $\mathsf{N}$ quadratic
equations:%
\begin{equation}
x_{n}\sum_{a=1}^{\mathsf{N}}g_{a}(\zeta _{n}^{(1)})x_{a}+x_{n}f(\zeta
_{n}^{(1)})=q_{n},\text{ \ \ \ }\forall n\in \{1,...,\mathsf{N}\},
\end{equation}%
once we define $x_{a}=\tau (\zeta _{a}^{(0)})$ for any $a\in \{1,...,\mathsf{%
N}\}$ and we write $\tau (\zeta _{n}^{(1)})$ in terms of the $x_{a}$. Thus we show that
 $$\tau (\lambda )=\left( Z_{Q}(\lambda )+F(\lambda )\right)
/Q(\lambda )\in \Sigma _{\mathcal{T}},$$
  completing the proof of
the theorem.
\end{proof}

\subsection{Completeness of the Bethe ansatz equations}

In the previous section we have shown that to solve the transfer matrix
spectral problem associated to the most general representations of the
trigonometric 6-vertex reflection algebra we have just to classify the set
of functions $Q(\lambda )$ of the form $\left( \ref{Q-form2}\right) $ for
which $\left( Z_{Q}(\lambda )+F(\lambda )\right) /Q(\lambda )$ is a
trigonometric polynomial; i.e. the set of functions $\Sigma _{\mathcal{Q}}$
completely fixes the set $%
\Sigma _{\mathcal{T}}$. We can show now that the previous characterization
of the transfer matrix spectrum allows to prove that $\Sigma _{InBAE}\subset \mathbb{C}^{\mathsf{N}}$ the set of all the solutions of inhomogeneous Bethe equations 
 $$ \{\lambda _{1},...,\lambda _{\mathsf{N}}\}\in \Sigma _{InBAE}$$
 if 
\begin{equation}
\mathbf{A}(\lambda _{a})Q_{\mathbf{\lambda }%
}(\lambda _{a}-\eta )+\mathbf{A}(-\lambda _{a})Q_{\mathbf{\lambda }}(\lambda
_{a}+\eta )=-F(\lambda _{a}),\text{ \ }\forall a\in \{1,...,\mathsf{N}\}%
 ,  \label{I-BAE}
\end{equation}%
 defines the complete set of transfer matrix\ eigenvalues. In
particular, the following corollary follows:

\begin{corollary}
\label{Theo-InBAE} Let the inhomogeneities $\{\xi _{1},...,\xi _{\mathsf{N}%
}\}\in \mathbb{C}$ $^{\mathsf{N}}$ be generic $\left( \ref{xi-conditions}%
\right) $ and let the boundary parameters $(\tau _{+},\alpha _{+},\beta _{+},\tau _{-},\alpha
_{-},\beta _{-})\in \mathbb{C}^{6}\backslash N_{SOV}$ satisfy $\left( \ref%
{Inhom-cond-BaxEq}\right) $ then 
 $\tau (\lambda )\in \Sigma _{\mathcal{T}}$ \ if and only if $\exists
!\{\lambda _{1},...,\lambda _{\mathsf{N}}\}\in \Sigma _{InBAE}$ such that:%
\begin{equation}
\tau (\lambda )=\frac{Z_{Q}(\lambda )+F(\lambda )}{Q(\lambda )}\text{ \ \
with \ \ }Q(\lambda )=2^{\mathsf{N}}\prod_{a=1}^{\mathsf{N}}\left( \cosh
2\lambda -\cosh 2\lambda _{a}\right) .
\end{equation}%
Moreover, under the condition of normality defined in Proposition \ref%
{normality}, the set $\Sigma _{InBAE}$ of all the solutions to the
inhomogeneous system of Bethe equations $\left( \ref{I-BAE}%
\right) $ contains $2^{\mathsf{N}}$ elements.
\end{corollary}

\section{Homogeneous Baxter equation}

\subsection{Boundary conditions annihilating the inhomogeneity of the Baxter
equation}

The description presented in the previous sections can be applied to completely
general integrable boundary terms including as a particular
case the boundary conditions for which the inhomogeneous term in the
functional Baxter equation vanishes. As these are still quite general
boundary conditions it is interesting to point out how the previous general
results explicitly look like in these cases.

\begin{theorem}
\label{homogeneousBE_N}
Let $(\tau _{+},\alpha _{+},\beta _{+},\tau _{-},\alpha _{-},\beta _{-})\in 
\mathbb{C}^{6}\backslash N_{SOV}$ satisfying the condition:%
\begin{equation}
\kappa _{+}\neq 0,\kappa _{-}\neq 0,\text{ \ }\exists i\in \left\{
0,1\right\} \text{\ }:Y^{(i,2\mathsf{N})}(\tau _{\pm },\alpha _{\pm },\beta
_{\pm })=0  \label{ond-homo-boundary}
\end{equation}%
and let the inhomogeneities $\{\xi _{1},...,\xi _{\mathsf{N}}\}\in \mathbb{C}
$ $^{\mathsf{N}}$ be generic \rf{xi-conditions}, then 
 $\tau (\lambda )\in \Sigma _{\mathcal{T}}$  if and only if $ \exists
!Q(\lambda )\in \Sigma _{\mathcal{Q}}$ such that
\begin{equation}
\tau (\lambda
)Q(\lambda )=\mathbf{A}(\lambda )Q(\lambda -\eta )+\mathbf{A}(-\lambda
)Q(\lambda +\eta ).
\end{equation}%
Or equivalently, $\tau (\lambda )\in \Sigma _{\mathcal{T}}$ \ if and only if 
$\exists !\{\lambda _{1},...,\lambda _{\mathsf{N}}\}\in \Sigma _{BAE}$ such
that:%
\begin{equation}
\tau (\lambda )=\frac{\mathbf{A}(\lambda )Q(\lambda -\eta )+\mathbf{A}%
(-\lambda )Q(\lambda +\eta )}{Q(\lambda )}\text{ \ \ with \ \ }Q(\lambda
)=2^{\mathsf{N}}\prod_{a=1}^{\mathsf{N}}\left( \cosh 2\lambda -\cosh
2\lambda _{a}\right) .
\end{equation}%
where:%
\begin{equation}
\Sigma _{BAE}=\left\{ \{\lambda _{1},...,\lambda _{\mathsf{N}}\}\in \mathbb{C%
}^{\mathsf{N}}:\mathbf{A}(\lambda _{a})Q_{\mathbf{\lambda }}(\lambda
_{a}-\eta )+\mathbf{A}(-\lambda _{a})Q_{\mathbf{\lambda }}(\lambda _{a}+\eta
)=0,\text{ \ }\forall a\in \{1,...,\mathsf{N}\}\right\} .  \label{BAE}
\end{equation}%
Moreover, under the condition of normality defined in Proposition \ref%
{normality},  the set $\Sigma _{BAE}$ of the solutions to the
homogeneous system of Bethe ansatz type equations $\left( \ref{BAE}\right) $
contains $2^{\mathsf{N}}$ elements.
\end{theorem}

\begin{proof}[Proof]
This theorem is just a rewriting of the results presented in the Theorem %
\ref{T-eigenvalue-F-eq} and Corollary  \ref{Theo-InBAE} for the case of vanishing
inhomogeneous term. Indeed  if the
conditions $\left( \ref{BAE}\right) $ are satisfied then automatically the conditions of the main theorem $\left( \ref%
{Inhom-cond-BaxEq}\right) $ are satisfied too that implies that the map from
the $\tau (\lambda )\in \Sigma _{\mathcal{T}}$ \ to the $\{\lambda
_{1},...,\lambda _{\mathsf{N}}\}\in \Sigma _{BAE}$ is indeed an isomorphism.
\end{proof}

\subsection{More general boundary conditions compatibles with homogeneous
Baxter equations}

We address here the problem of describing the boundary conditions:%
\begin{equation}
\kappa _{+}\neq 0,\kappa _{-}\neq 0,\text{ \ }\exists i\in \left\{
0,1\right\} ,\mathsf{M}\in \left\{ 0,...,\mathsf{N}-1\right\} :Y^{(i,2%
\mathsf{M})}(\tau _{\pm },\alpha _{\pm },\beta _{\pm })=0,
\label{Cond-homo-M}
\end{equation}%
for which the conditions $\left( \ref{Inhom-cond-BaxEq}\right) $ are not
satisfied and then the Theorem \ref{T-eigenvalue-F-eq} cannot be 
directly applied. In these $2\mathsf{N}$ hyperplanes in the space of the boundary
parameters we have just to modify this theorem to take into account that the
Baxter equation associated to the choice of coefficient $\mathbf{A}(\lambda
) $ is indeed compatible with the homogeneous Baxter equation for a
special choice of the polynomial $Q(\lambda )$. First we
define the following functions%
\begin{equation}
Q_{\mathsf{M}}(\lambda )=2^{\mathsf{M}}\prod_{b=1}^{\mathsf{M}}\left( \cosh
2\lambda -\cosh 2\lambda _{b}^{(\mathsf{M})}\right) .  \label{Q-form-M}
\end{equation}%
We introduce also the  set of polynomials $\Sigma _{\mathcal{Q}}^{\mathsf{M}}$ such that
$Q_{\mathsf{M}}(\lambda )\in \Sigma _{\mathcal{Q}}^{\mathsf{M}}$ if $Q_{\mathsf{M}}(\lambda )$ has a form $\left( \ref{Q-form-M}\right)$ and
\begin{equation*}
\tau (\lambda )=\frac{%
\mathbf{A}(\lambda )Q_{\mathsf{M}}(\lambda -\eta )+\mathbf{A}(-\lambda )Q_{%
\mathsf{M}}(\lambda +\eta )}{Q_{\mathsf{M}}(\lambda )}%
\end{equation*}%
is a trigonometric polynomial. Then we can  define the corresponding set  $\Sigma _{\mathcal{T}}^{\mathsf{M}}$%
\begin{equation}
\Sigma _{\mathcal{T}}^{\mathsf{M}}=\left\{ \tau (\lambda ):\tau (\lambda
)\equiv \frac{\mathbf{A}(\lambda )Q_{\mathsf{M}}(\lambda -\eta )+\mathbf{A}%
(-\lambda )Q_{\mathsf{M}}(\lambda +\eta )}{Q_{\mathsf{M}}(\lambda )}\text{
\, if  }Q_{\mathsf{M}}(\lambda )\in \Sigma _{\mathcal{Q}}^{\mathsf{M}}\right\}
.
\end{equation}%
It is simple to prove the validity of the following:

\begin{lemma}\label{mixed-condition}
Let the boundary conditions $\left( \ref{Cond-homo-M}\right) $ be satisfied,
then $\Sigma _{\mathcal{T}}^{\mathsf{M}}\subset\Sigma _{\mathcal{T}}$
and moreover for any $\tau (\lambda )\in \Sigma _{\mathcal{T}}^{\mathsf{M}}$
there exists one and only one $Q_{\mathsf{M}}(\lambda )\in \Sigma _{\mathcal{%
Q}}^{\mathsf{M}}$ such that:%
\begin{equation}\label{homogen-Bax-eq-M}
\tau (\lambda )Q_{\mathsf{M}}(\lambda )=\mathbf{A}(\lambda )Q_{\mathsf{M}%
}(\lambda -\eta )+\mathbf{A}(-\lambda )Q_{\mathsf{M}}(\lambda +\eta ),
\end{equation}%
and for any $\tau (\lambda )\in \Sigma _{\mathcal{T}}\backslash \Sigma _{%
\mathcal{T}}^{\mathsf{M}}$ there exists one and only one $Q(\lambda )\in
\Sigma _{\mathcal{Q}}$ such that:%
\begin{equation}\label{inhomogen-Bax-eq-M}
\tau (\lambda )Q(\lambda )=\mathbf{A}(\lambda )Q(\lambda -\eta )+\mathbf{A}%
(-\lambda )Q(\lambda +\eta )+F(\lambda ).
\end{equation}
\end{lemma}

\begin{proof}[Proof]
The proof follows the one given for the main Theorem \ref{T-eigenvalue-F-eq}
we have just to observe that thanks to the boundary conditions $\left( \ref%
{Cond-homo-M}\right) $ the set $\Sigma _{\mathcal{T}}^{\mathsf{M}}$ is
formed by transfer matrix eigenvalues as the Baxter equation implies that
for any $\tau (\lambda )\in \Sigma _{\mathcal{T}}^{\mathsf{M}}$ the systems
of equations $\left( \ref{System-A}\right) ,$ $\left( \ref{System-B}\right) $
and $\left( \ref{System1}\right) $ are satisfied and moreover that the
asymptotics of the $\tau (\lambda )\in \Sigma _{\mathcal{T}}^{\mathsf{M}}$
is exactly that of the transfer matrix eigenvalues.
\end{proof}

Finally, it is interesting to remark that under the boundary conditions $%
\left( \ref{Inhom-cond-BaxEq}\right) $  the complete
characterization of the spectrum of the transfer matrix is given in terms of
the even polynomials $Q(\lambda )$ all of fixed degree $2\mathsf{N}$ and
form $\left( \ref{Q-form2}\right) $ which are solutions of the
inhomogeneous/homogeneous Baxter equation. However, in the cases  when the boundary parameters satisfy the constraints $\left( \ref%
{Cond-homo-M}\right) $ for a given $\mathsf{M}\in \left\{ 0,...,\mathsf{N}%
-1\right\} $ a part of the transfer matrix spectrum can be
defined by polynomials of smaller degree; i.e. the $Q_{\mathsf{M}}(\lambda
)\in \Sigma _{\mathcal{Q}}^{\mathsf{M}}$ for the fixed $\mathsf{M}\in
\left\{ 0,...,\mathsf{N}-1\right\} $.

\section{Discrete symmetries and equivalent Baxter equations}
\label{sect-descretesym}

It is important to point out that we have some large amount of freedom in
the choice of the functional reformulation of the SOV characterization of
the transfer matrix spectrum. We have reduced it looking for trigonometric
polynomial solutions $Q(\lambda )$ of the second order difference equations with
coefficients $\mathbf{A}(\lambda )$ which are rational trigonometric
functions. It makes the finite difference terms $\mathbf{A}(\lambda
)Q(\lambda -\eta )+\mathbf{A}(-\lambda )Q(\lambda +\eta )$ in the functional
equation a trigonometric polynomial. Indeed, this assumption reduces the
possibility to use the following gauge transformations of the coefficients
allowed instead by the SOV characterization:%
\begin{equation}
\mathbf{A}_{\alpha }(\lambda )=\alpha (\lambda )\mathbf{A}(\lambda ),\text{
\ }\mathbf{D}_{\alpha }(\lambda )=\frac{\mathbf{A}(-\lambda )}{\alpha
(\lambda +\eta )}.
\end{equation}%
In the following we discuss simple transformations that do not modify the
 functional form of the coefficients   allowing  equivalent
reformulations of the SOV spectrum by Baxter equations.

\subsection{Discrete symmetries of the transfer matrix spectrum}

It is not difficult to see that the spectrum (eigenvalues) of the transfer matrix
presents the following invariance:

\begin{lemma}
\label{Lem-invariance}We denote explicitly the dependence from the
boundary parameters in the set of boundary parameters $\Sigma _{\mathcal{T}}^{(\tau _{+},\alpha
_{+},\beta _{+},\tau _{-},\alpha _{-},\beta _{-})}$ of the eigenvalue
functions of the transfer matrix $\mathcal{T}(\lambda )$, then this set is
invariant under the following $Z_{2}^{\otimes 3}$ transformations of the
boundary parameters:%
\begin{align}
&\Sigma _{\mathcal{T}}^{(\tau _{+},\alpha _{+},\beta _{+},\tau _{-},\alpha
_{-},\beta _{-})}\equiv \Sigma _{\mathcal{T}}^{(\epsilon _{\tau }\tau
_{+},\epsilon _{\alpha }\alpha _{+},\epsilon _{\beta }\beta _{+},\epsilon
_{\tau }\tau _{-},\epsilon _{\alpha }\alpha _{-},\epsilon _{\beta }\beta
_{-})}\ \\ &
\forall (\epsilon _{\tau },\epsilon _{\alpha },\epsilon _{\beta
})\in \{-1,1\}\times \{-1,1\}\times \{-1,1\}.\nonumber
\end{align}
\end{lemma}

\begin{proof}[Proof]
To prove this statement it is enough to look at the SOV characterization
which defines completely the transfer matrix spectrum, i.e.
the set $\Sigma _{\mathcal{T}}$, and to prove that it is invariant under the
above considered $Z_{2}^{\otimes 3}$ transformations of the boundary
parameters. We have first to remark that the central values $%
\left( \ref{Central-asymp}\right) $-$\left( \ref{Central-2}\right) $ of
the transfer matrix $\mathcal{T}(\lambda )$ are invariant under these
discrete transformations and then the function $f(\lambda )$, defined in $%
\left( \ref{f-function}\right) $, is invariant too and the same is true for
the form $\left( \ref{Interpolation-Form-T}\right) $ of the interpolation
polynomial describing the elements of $\Sigma _{\mathcal{T}}$. Then the invariance of
 the SOV characterization $\left( \ref{Quadratic System}\right) $\ follows from the
invariance of the quantum determinant%
\begin{eqnarray}
\mathrm{det}_{q}K_{+}(\lambda )\mathrm{det}_{q}\,\mathcal{U}_{-}(\lambda )
&=&\sinh (2\eta -2\lambda )\sinh (2\lambda +2\eta )g_{+}(\lambda +\eta
/2)g_{+}(-\lambda +\eta /2)g_{-}(\lambda +\eta /2)  \notag \\
&&\times g_{-}(-\lambda +\eta /2)a(\lambda +\eta /2)d(\lambda -\eta
/2)a(-\lambda +\eta /2)d(-\lambda -\eta /2)
\end{eqnarray}%
 under these discrete transformations.
\end{proof}

It is important to underline that the above $Z_{2}^{\otimes 3}$ transformations
of the boundary parameters do indeed change the transfer matrix $\mathcal{T}%
(\lambda )$ and the Hamiltonian and so this invariance  is
equivalent to the statement that these different transfer matrices are all
isospectral. In particular, it is simple to find the similarity matrices
implementing the following $Z_{2}$ transformations of the boundary
parameters:%
\begin{eqnarray}
\mathcal{T}(\lambda |-\tau _{+},-\zeta _{+},\kappa _{+},-\tau _{-},-\zeta
_{-},\kappa _{-}) &=&\Gamma _{y}\mathcal{T}(\lambda |\tau _{+},\zeta
_{+},\kappa _{+},\tau _{-},\zeta _{-},\kappa _{-})\Gamma _{y},\text{ \ \ \ }%
\Gamma _{y}\equiv \otimes _{n=1}^{\mathsf{N}}\sigma _{n}^{y}, \\
\mathcal{T}(\lambda |\tau _{+},\zeta _{+},-\kappa _{+},\tau _{-},\zeta
_{-},-\kappa _{-}) &=&\Gamma _{z}\mathcal{T}(\lambda |\tau _{+},\zeta
_{+},\kappa _{+},\tau _{-},\zeta _{-},\kappa _{-})\Gamma _{z},\text{ \ \ \ }%
\Gamma _{z}\equiv \otimes _{n=1}^{\mathsf{N}}\sigma _{n}^{z}.
\end{eqnarray}

\subsection{Equivalent Baxter equations and the SOV spectrum}

The invariance of the spectrum $\Sigma _{\mathcal{T}}$ under these $Z_{2}^{\otimes 3}$ transformations of the boundary parameters
can be used to define equivalent Baxter equation reformulation of $\Sigma _{%
\mathcal{T}}$. More precisely, let us introduce the following functions $%
\mathbf{A}_{(\epsilon _{\tau },\epsilon _{\alpha },\epsilon _{\beta
})}(\lambda )$ and $F_{(\epsilon _{\tau },\epsilon _{\alpha },\epsilon
_{\beta })}(\lambda )$ obtained respectively by implementing the $%
Z_{2}^{\otimes 3}$ transformations:%
\begin{equation}
(\tau _{+},\alpha _{+},\beta _{+},\tau _{-},\alpha _{-},\beta
_{-})\rightarrow (\epsilon _{\tau }\tau _{+},\epsilon _{\alpha }\alpha
_{+},\epsilon _{\beta }\beta _{+},\epsilon _{\tau }\tau _{-},\epsilon
_{\alpha }\alpha _{-},\epsilon _{\beta }\beta _{-}),
\end{equation}%
then the following characterizations hold for any fixed $(\epsilon _{\tau
},\epsilon _{\alpha },\epsilon _{\beta })\in \{-1,1\}\times \{-1,1\}\times
\{-1,1\}$:

\begin{theorem}
\label{T-eigenvalue-F-eq-gen}Let the inhomogeneities $\{\xi _{1},...,\xi _{%
\mathsf{N}}\}\in \mathbb{C}$ $^{\mathsf{N}}$ be generic  \rf{xi-conditions} and let the boundary parameters $(\tau _{+},\alpha _{+},\beta _{+},\tau _{-},\alpha _{-},\beta
_{-})\in \mathbb{C}^{6}\backslash N_{SOV}$ satisfy the following conditions:%
\begin{equation}
\kappa _{+}\neq 0,\kappa _{-}\neq 0,\text{ \ }Y^{(i,2r)}(\epsilon _{\tau
}\tau _{\pm },\epsilon _{\alpha }\alpha _{\pm },\epsilon _{\beta }\beta
_{\pm })\neq 0\text{ \ }\forall i\in \left\{ 0,1\right\} ,r\in \left\{ 0,...,%
\mathsf{N}-1\right\} ,  \label{Inhom-cond-BaxEq-gen}
\end{equation}%
then 
$\tau (\lambda )\in \Sigma _{\mathcal{T}}$  if and only if $\exists
!Q(\lambda )\in \Sigma _{\mathcal{Q}}$ such that 
\begin{equation}
\tau (\lambda
)Q(\lambda )=Z_{Q,(\epsilon _{\tau },\epsilon _{\alpha },\epsilon _{\beta
})}(\lambda )+F_{(\epsilon _{\tau },\epsilon _{\alpha },\epsilon _{\beta
})}(\lambda ),
\end{equation}%
where:%
\begin{equation}
Z_{Q,(\epsilon _{\tau },\epsilon _{\alpha },\epsilon _{\beta })}(\lambda )=%
\mathbf{A}_{(\epsilon _{\tau },\epsilon _{\alpha },\epsilon _{\beta
})}(\lambda )(\lambda )Q(\lambda -\eta )+\mathbf{A}_{(\epsilon _{\tau
},\epsilon _{\alpha },\epsilon _{\beta })}(-\lambda )Q(\lambda +\eta ).
\end{equation}
\end{theorem}

\begin{proof}[Proof]
The proof follows step by step the one given for the main Theorem \ref{T-eigenvalue-F-eq}.
\end{proof}

\subsection{General validity of the
inhomogeneous Baxter equations}

The previous reformulations of the spectrum in terms of different
inhomogeneous Baxter equations and the observation that the conditions under
which the Theorem does not apply are related to the choice of the $(\epsilon
_{\tau },\epsilon _{\alpha },\epsilon _{\beta })\in \{-1,1\}\times
\{-1,1\}\times \{-1,1\}$ allow us to prove that unless the boundary
parameters are lying on a finite lattice of step $\eta $ we can always use
an inhomogeneous Baxter equations\ to completely characterize the spectrum
of the transfer matrix. More precisely, let us introduce the following
hyperplanes in the space of the boundary parameters:%
\begin{equation}
M\equiv \left\{ 
\begin{array}{l}
(\tau _{+},\alpha _{+},\beta _{+},\tau _{-},\alpha _{-},\beta _{-})\in 
\mathbb{C}^{6}:\exists (r_{+,+},r_{-,+},r_{-,-})\in \{0,...,\mathsf{N}-1\}
\\ 
\text{ such that: \ \ \ }\left\{ 
\begin{array}{l}
r_{+,+}+r_{-,-}-r_{-,+}\in \{0,...,\mathsf{N}-1\} \\ 
\alpha _{+}+\alpha _{-}=(r_{-,+}-r_{+,+})\eta \\ 
\beta _{-}-\beta _{+}=(r_{-,-}-r_{-,+})\eta \\ 
\tau _{-}-\tau _{+}=(\mathsf{N}-1+r_{-,-}-3r_{+,+})\eta%
\end{array}%
\right.%
\end{array}%
\right\}  \label{Def-M}
\end{equation}%
then the following theorem holds:

\begin{theorem}
Let the inhomogeneities $\{\xi _{1},...,\xi _{\mathsf{N}}\}\in \mathbb{C}$ $%
^{\mathsf{N}}$ satisfy the conditions \rf{xi-conditions} and let $(\tau _{+},\alpha _{+},\beta _{+},\tau _{-},\alpha _{-},\beta
_{-})\in \mathbb{C}^{6}\backslash \left( M\cup N_{SOV}\right) $ then we can
always find a $(\epsilon _{\tau },\epsilon _{\alpha },\epsilon _{\beta })\in
\{-1,1\}\times \{-1,1\}\times \{-1,1\}$ such that $\tau (\lambda )\in\Sigma _{\mathcal{T}}$ 
  if and only if $\exists
!Q(\lambda )\in \Sigma _{\mathcal{Q}}$ such that 
\begin{equation}
\tau (\lambda
)Q(\lambda )=Z_{Q,(\epsilon _{\tau },\epsilon _{\alpha },\epsilon _{\beta
})}(\lambda )+F_{(\epsilon _{\tau },\epsilon _{\alpha },\epsilon _{\beta
})}(\lambda ).
\end{equation}
\end{theorem}

\begin{proof}[Proof]
The Theorem \ref{T-eigenvalue-F-eq-gen} does not apply if $\exists i\in
\left\{ 0,1\right\} $ and $\exists r\in \left\{ 0,...,\mathsf{N}-1\right\} $
such that the following system of conditions on the boundary parameters are
satisfied:%
\begin{equation}
Y^{(i,2r)}(\epsilon _{\tau }\tau _{\pm },\epsilon _{\alpha }\alpha _{\pm
},\epsilon _{\beta }\beta _{\pm })=0\text{ \ }\forall (\epsilon _{\tau
},\epsilon _{\alpha },\epsilon _{\beta })\in \{-1,1\}^{\otimes 3}
\label{Cond-general-homo}
\end{equation}
then by simple computations it is possible to observe that the set $M$
defined in $\left( \ref{Def-M}\right) $ indeed coincides with the following
set:%
\begin{equation}
\left\{ (\tau _{+},\alpha _{+},\beta _{+},\tau _{-},\alpha _{-},\beta
_{-})\in \mathbb{C}^{6}:\text{ }\exists i\in \left\{ 0,1\right\} ,r\in
\left\{ 0,...,\mathsf{N}-1\right\} \text{ such that }\left( \ref%
{Cond-general-homo}\right) \text{ is satisfied}\right\} ,
\end{equation}%
from which the theorem clearly follows.
\end{proof}

\subsection{Homogeneous
Baxter equation}

The discrete symmetries of the transfer matrix  allow also
to define the general conditions on the boundary parameters for which 
 the spectrum can be characterized by a homogeneous Baxter equation. In
particular the following corollary holds:

\begin{corollary}
Let $(\tau _{+},\alpha _{+},\beta _{+},\tau _{-},\alpha _{-},\beta _{-})\in 
\mathbb{C}^{6}\backslash N_{SOV}$ satisfy the condition:%
\begin{align}
&\kappa _{+}\neq 0,\kappa _{-}\neq 0,\nonumber\\
&\exists i\in \left\{
0,1\right\} ,\text{\ }\exists (\epsilon _{\tau },\epsilon _{\alpha
},\epsilon _{\beta })\in \{-1,1\}\times \{-1,1\}\times \{-1,1\}:Y^{(i,2%
\mathsf{N})}(\epsilon _{\tau }\tau _{\pm },,\epsilon _{\alpha }\alpha _{\pm
},\epsilon _{\beta }\beta _{\pm })=0
\end{align}%
and let the inhomogeneities $\{\xi _{1},...,\xi _{\mathsf{N}}\}\in \mathbb{C}
$ $^{\mathsf{N}}$ be generic \rf{xi-conditions}, then $\tau (\lambda )\in \Sigma _{\mathcal{T}}$ if and only if $\exists
!Q(\lambda )\in \Sigma _{\mathcal{Q}}$ such that 
\begin{equation}
\tau (\lambda
)Q(\lambda )=\mathbf{A}_{(\epsilon _{\tau },\epsilon _{\alpha },\epsilon
_{\beta })}(\lambda )(\lambda )Q(\lambda -\eta )+\mathbf{A}_{(\epsilon
_{\tau },\epsilon _{\alpha },\epsilon _{\beta })}(-\lambda )Q(\lambda +\eta
).
\end{equation}%
Or equivalently we can define the set of all the solutions of the Bethe equations
\begin{equation}
\Sigma _{BAE}=\left\{ \{\lambda _{1},...,\lambda _{\mathsf{N}}\}\in \mathbb{C%
}^{\mathsf{N}}:\mathbf{A}(\lambda _{a})Q_{\mathbf{\lambda }}(\lambda
_{a}-\eta )+\mathbf{A}(-\lambda _{a})Q_{\mathbf{\lambda }}(\lambda _{a}+\eta
)=0,\text{ \ }\forall a\in \{1,...,\mathsf{N}\}\right\} .
\end{equation}%
Then $\tau (\lambda )\in \Sigma _{\mathcal{T}}$ \ if and only if 
$\exists !\{\lambda _{1},...,\lambda _{\mathsf{N}}\}\in \Sigma _{BAE}$ such
that:%
\begin{equation}
\tau (\lambda )=\frac{\mathbf{A}_{(\epsilon _{\tau },\epsilon _{\alpha
},\epsilon _{\beta })}(\lambda )Q(\lambda -\eta )+\mathbf{A}_{(\epsilon
_{\tau },\epsilon _{\alpha },\epsilon _{\beta })}(-\lambda )Q(\lambda +\eta )%
}{Q(\lambda )},
\end{equation}%
with
$$Q(\lambda )=2^{\mathsf{N}}\prod_{a=1}^{%
\mathsf{N}}\left( \cosh 2\lambda -\cosh 2\lambda _{a}\right).$$
Moreover, under the condition of normality defined in Proposition \ref%
{normality},  the set $\Sigma _{BAE}$ of the solutions to the
homogeneous system of Bethe ansatz type equations $\left( \ref{BAE}\right) $
contains $2^{\mathsf{N}}$ elements.
\end{corollary}

\section{XXX  chain by SOV and Baxter equation}

The construction of the SOV\ characterization can be naturally applied 
 in the case of the rational 6-vertex $R$-matrix, which in the
homogeneous limit reproduces the XXX open quantum spin-1/2 chain with
general integrable boundary conditions\footnote{%
Here we use notations similar to those introduced in the papers \cite{CaoJYW2013} and
\cite{Nep-2013} where some inhomogeneous Baxter equation ansatzs appear with the
aim to make simpler for the reader a comparison when the limit of
homogeneous chain is implemented.}. Let us define:%
\begin{equation}
R_{12}(\lambda )=\left( 
\begin{array}{cccc}
\lambda +\eta & 0 & 0 & 0 \\ 
0 & \lambda & \eta & 0 \\ 
0 & \eta & \lambda & 0 \\ 
0 & 0 & 0 & \lambda +\eta%
\end{array}%
\right) \in \text{End}(\mathcal{H}_{1}\otimes \mathcal{H}_{2}).
\end{equation}%
Due to the $SU(2)$ invariance of the bulk monodromy matrix the boundary
matrices defining the most general integrable boundary conditions can be
always recasted in the following form:%
\begin{equation}
K_{-}(\lambda ;p)=\left( 
\begin{array}{cc}
\lambda -\eta /2+p & 0 \\ 
0 & p-\lambda +\eta /2%
\end{array}%
\right) ,\text{ \ \ \ }K_{+}(\lambda ;q,\xi )=\left( 
\begin{array}{cc}
\lambda +\eta /2+q & \xi (\lambda +\eta /2) \\ 
\xi (\lambda +\eta /2) & q-(\lambda +\eta /2)%
\end{array}%
\right) ,
\end{equation}%
leaving only three arbitrary complex parameters here denoted with $\xi ,$ $p$
and $q$. Then the one parameter family of commuting transfer matrices:%
\begin{equation}
\mathcal{T}(\lambda )=\text{tr}_{0}\{K_{+}(\lambda )\,M(\lambda
)\,K_{-}(\lambda )\hat{M}(\lambda )\}\in \text{\thinspace End}(\mathcal{H}),
\end{equation}%
in the homogeneous limit leads to the following Hamiltonian:%
\begin{equation}
H=\sum_{n=1}^{\mathsf{N}}\left( \sigma _{n}^{x}\sigma _{n+1}^{x}+\sigma
_{n}^{y}\sigma _{n+1}^{y}+\sigma _{n}^{z}\sigma _{n+1}^{z}\right) +\frac{%
\sigma _{\mathsf{N}}^{z}}{p}+\frac{\sigma _{1}^{z}+\xi \sigma _{1}^{x}}{q}.
\end{equation}%
It is simple to show that the following identities hold:%
\begin{equation}
\mathrm{det}_{q}K_{+}(\lambda )\mathrm{det}_{q}\,\mathcal{U}_{-}(\lambda
)=4(\lambda ^{2}-\eta ^{2})(\lambda ^{2}-p^{2})((1+\xi ^{2})\lambda
^{2}-q^{2})\prod_{b=1}^{\mathsf{N}}(\lambda ^{2}-(\xi _{n}+\eta
)^{2})(\lambda ^{2}-(\xi _{n}-\eta )^{2}).
\end{equation}%
We define:%
\begin{equation}
\mathbf{A}(\lambda )=(-1)^{\mathsf{N}}\frac{2\lambda +\eta }{2\lambda }%
(\lambda -\eta /2+p)(\sqrt{(1+\xi ^{2})}(\lambda -\eta /2)+q)\prod_{b=1}^{%
\mathsf{N}}(\lambda -\zeta _{b}^{(0)})(\lambda +\zeta _{b}^{(1)}),
\end{equation}%
then it is easy to derive the following quantum determinant identity:%
\begin{equation}
\frac{\mathrm{det}_{q}K_{+}(\lambda )\mathrm{det}_{q}\,\mathcal{U}_{-}(\lambda )}{%
(4\lambda ^{2}-\eta ^{2})}=\mathbf{A}(\lambda +\eta /2)\mathbf{A}(-\lambda
+\eta /2).
\end{equation}%
 From the form of the boundary matrices it is clear that for the rational
6-vertex case one can directly derive the SOV representations  using the
method developed in \cite{Nic12b} without any need to introduce Baxter's
gauge transformations. Some results in this case also appeared in 
\cite{FraSW08,FraGSW11} based on a functional
version of the separation of variables of Sklyanin, a method which allows to
define the eigenvalues and wave-functions but which does not allow to
construct in the original Hilbert space of the quantum chain the transfer
matrix eigenstates. 

The separation of variable description in this
rational 6-vertex case reads:

\begin{theorem}
\label{C:T-eigenstates- copy(1)}Let the inhomogeneities $\{\xi _{1},...,\xi
_{\mathsf{N}}\}\in \mathbb{C}$ $^{\mathsf{N}}$  be generic:
\begin{equation}
\xi _{a}\neq \pm\xi _{b}+r\eta \text{ \ }\forall a\neq b\in \{1,...,\mathsf{N}%
\}\,\,\text{and\thinspace \thinspace }r\in \{-1,0,1\},  \label{xi-conditions-xxx}
\end{equation}
then $\mathcal{T}(\lambda )$ has simple spectrum and $\Sigma _{\mathcal{T}}$
is characterized by:%
\begin{equation}
\Sigma _{\mathcal{T}}=\left\{ \tau (\lambda ):\tau (\lambda )=f(\lambda
)+\sum_{a=1}^{\mathsf{N}}g_{a}(\lambda )x_{a},\text{ \ \ }\forall
\{x_{1},...,x_{\mathsf{N}}\}\in \Sigma _{T}\right\} ,
\end{equation}%
where:%
\begin{equation}
g_{a}(\lambda )=\frac{4\lambda ^{2}-\eta ^{2}}{4{\zeta
_{a}^{(0)}} ^{2}-\eta ^{2}}\,\prod_{\substack{ b=1  \\ b\neq a}}^{%
\mathsf{N}}\frac{\lambda ^{2}-{ \zeta _{b}^{(0)}} ^{2}}{{
\zeta _{a}^{(0)}} ^{2}-{ \zeta _{b}^{(0)}} ^{2}}\quad \text{
\ for }a\in \{1,...,\mathsf{N}\},
\end{equation}%
and%
\begin{equation}
f(\lambda )=\prod_{b=1}^{\mathsf{N}}\frac{\lambda ^{2}-{ \zeta
_{b}^{(0)}}^{2}}{{ \zeta _{a}^{(0)}} ^{2}-{\zeta
_{b}^{(0)}} ^{2}}\mathbf{A}(\eta /2)+2\left( 4\lambda ^{2}-\eta
^{2}\right) \,\prod_{b=1}^{\mathsf{N}}\lambda ^{2}-{ \zeta
_{b}^{(0)}} ^{2},
\end{equation}%
$\Sigma _{T}$ is the set of solutions to the following inhomogeneous system
of $\mathsf{N}$ quadratic equations:%
\begin{equation}
x_{n}\sum_{a=1}^{\mathsf{N}}g_{a}(\zeta _{n}^{(1)})x_{a}+x_{n}f(\zeta
_{n}^{(1)})=q_{n},\text{ \ \ \ }q_{n}=\frac{\mathrm{det}_{q}K_{+}(\xi _{n})%
\mathrm{det}_{q}\,\mathcal{U}_{-}(\xi _{n})}{\eta -4\xi _{n}^{2}},\text{ \ \ }%
\forall n\in \{1,...,\mathsf{N}\},
\end{equation}%
in $\mathsf{N}$ unknowns $\{x_{1},...,x_{\mathsf{N}}\}$.
\end{theorem}

We are now ready to present the following equivalent characterization of the
transfer matrix spectrum:

\begin{theorem}
\label{T-eigenvalue-F-eq copy(1)}Let the inhomogeneities $\{\xi _{1},...,\xi
_{\mathsf{N}}\}\in \mathbb{C}$ $^{\mathsf{N}}$ be generic \rf{xi-conditions-xxx},
then for $\xi \neq 0$ the
set of transfer matrix eigenvalue functions $\Sigma _{\mathcal{T}}$ is
characterized by:%
\begin{equation*}
\tau (\lambda )\in \Sigma _{\mathcal{T}}\text{ \ if and only if }\exists
!Q(\lambda )=\prod_{b=1}^{\mathsf{N}}\left( \lambda ^{2}-\lambda
_{b}^{2}\right) \text{ such that }\tau (\lambda )Q(\lambda )=Z_{Q}(\lambda
)+F(\lambda ),
\end{equation*}%
with%
\begin{equation}
F(\lambda )=2(1-\sqrt{(1+\xi ^{2})})\left( 4\lambda ^{2}-\eta ^{2}\right)
\,\prod_{b=1}^{\mathsf{N}}\prod_{i=0}^{1}\left( \lambda ^{2}- {�\zeta
_{b}^{(i)}} ^{2}\right) .
\end{equation}
\end{theorem}

\begin{proof}[Proof]
The proof presented in Theorem \ref{T-eigenvalue-F-eq} applies with small
modifications also to present rational case.
\end{proof}

The previous characterization of the transfer matrix spectrum allows to
prove that the set $\Sigma _{InBAE}\subset\mathbb{C}^\mathsf{N}$ of all the solutions of the Bethe equations%
$$ \{\lambda _{1},...,\lambda _{\mathsf{N}}\}\in\Sigma _{InBAE}$$
if
\begin{equation}
\mathbf{A}(\lambda _{a})Q_{\mathbf{\lambda }%
}(\lambda _{a}-\eta )+\mathbf{A}(-\lambda _{a})Q_{\mathbf{\lambda }}(\lambda
_{a}+\eta )=-F(\lambda _{a}),\text{ \ }\forall a\in \{1,...,\mathsf{N}%
\} ,  \label{I-BAE-XXX}
\end{equation}%
 define the complete set of transfer matrix eigenvalues. In
particular, the following corollary can be proved:

\begin{corollary}
\label{Theo-InBAE-XXX}Let the inhomogeneities $\{\xi _{1},...,\xi _{\mathsf{N%
}}\}\in \mathbb{C}$ $^{\mathsf{N}}$ satisfy the following conditions \rf{xi-conditions},
then $\mathcal{T}(\lambda )$ has simple spectrum and for $\xi \neq 0$ then $%
\tau (\lambda )\in \Sigma _{\mathcal{T}}$ \ if and only if $\exists
!\{\lambda _{1},...,\lambda _{\mathsf{N}}\}\in \Sigma _{InBAE}$ such that:%
\begin{equation}
\tau (\lambda )=\frac{Z_{Q}(\lambda )+F(\lambda )}{Q(\lambda )}\text{ \ \
with \ \ }Q(\lambda )=\prod_{b=1}^{\mathsf{N}}\left( \lambda ^{2}-\lambda
_{b}^{2}\right) .
\end{equation}
\end{corollary}

\section{Homogeneous chains and existing numerical analysis}

It is important to stress that the spectrum construction together with the
corresponding statements of completeness presented in this paper strictly
work for the most general spin 1/2 representations of the 6-vertex
reflection algebra only for generic inhomogeneous chains. However, it is
worth mentioning that the transfer matrix as well as the coefficients and
the inhomogeneous term in our functional equation characterization of the
SOV spectrum are analytic functions of the inhomogeneities $\{\xi _{j}\}$ so
we can take without any problem the homogeneous limit ($\xi _{a}\rightarrow
0$ $\forall a\in \{1,...,\mathsf{N}\}$) in the functional equations.
The main problem to be addressed then is the completeness of the 
description by this functional equations. 
Some first understanding of this central question can be derived
looking at the numerical analysis \cite{Nep-R-2003,Nep-R-2003add} of the completeness of Bethe Ansatz
equations 
when the boundary constraints are satisfied and for the open XXX chain with general boundary
terms \cite{Nep-2013}.

\subsection{Comparison with numerical results for the XXZ chain}

The numerical checks of the completeness of Bethe Ansatz equations for the open XXZ quantum spin 1/2 chains  were first done in 
\cite{Nep-R-2003} 
for the chains with   non-diagonal boundaries satisfying boundary constraints: 
\begin{equation}
\kappa _{+}\neq 0,\kappa _{-}\neq 0,\text{ \ }\exists i\in \left\{
0,1\right\} ,\mathsf{M}\in \mathbb{N}\text{\ }:Y^{(i,2\mathsf{M})}(\tau _{\pm },\alpha _{\pm
},\beta _{\pm })=0.
\end{equation}%
Indeed, under these conditions some generalizations of algebraic Bethe
Ansatz can be used and so the corresponding Bethe equations can be defined.

In particular, the Nepomechie-Ravanini's numerical results reported in \cite%
{Nep-R-2003,Nep-R-2003add} suggest that the Bethe ansatz equations $\left( \ref{BAE}%
\right) $ in the homogeneous limit for the roots of the  $Q$ function:%
\begin{equation}
Q(\lambda )=2^{\mathsf{M}}\prod_{a=1}^{\mathsf{M}}\left( \cosh 2\lambda -\cosh 2\lambda
_{a}\right) ,
\end{equation}%
with the degree $\mathsf{M}$ obtained from the boundary constraint
\begin{itemize}
\item for $\mathsf{M}=\mathsf{N}$ they define the complete transfer matrix spectrum.
\item for $\mathsf{M}<\mathsf{N}$  the complete spectrum of the transfer matrix  contains two parts described by different Baxter equations. The first one has  trigonometric polynomial solutions of degree $2\mathsf{M}$ the second one has a trigonometric polynomial  solutions of degree $2\mathsf{N}-2-2\mathsf{M}$.
\item for $\mathsf{M}>\mathsf{N}$ the complete spectrum of the transfer matrix
spectrum plus $\tau (\lambda )$ functions which do not belong to the
spectrum of the transfer matrix.
\end{itemize}

These results seem to be compatible with our characterization for the
inhomogeneous chains. Indeed, the case $\mathsf{M}=\mathsf{N}$ coincides with the
case in which our Baxter functional equation becomes homogeneous. Theorem \ref{homogeneousBE_N} states
 that in this case for generic inhomogeneities the Bethe ansatz is complete so we can expect  (from the numerical analysis) that completeness will survive in the homogeneous limit.


In the case $\mathsf{M}<$\textsf{N}, our description of the spectrum by Lemma \ref%
{mixed-condition} separates the spectrum in two parts. A first part of the
spectrum is described by trigonometric polynomial solutions of degree $2\mathsf{M}$
to the homogeneous Baxter equation \rf{homogen-Bax-eq-M} and a second part is instead
described by trigonometric polynomial solutions of degree 2\textsf{N} of the
inhomogeneous Baxter equation \rf{inhomogen-Bax-eq-M}. However, by
implementing the following discrete symmetry transformations $\alpha _{\pm
}\rightarrow -\alpha _{\pm }$, $\beta _{\pm }\rightarrow -\beta _{\pm }$, $%
\tau _{\pm }\rightarrow -\tau _{\pm }$ and applying the same Lemma \ref%
{mixed-condition} w.r.t. the Baxter equations with coefficients $\mathbf{A}%
_{(-,-,-)}(\lambda )$ we get an equivalent description of the spectrum
separated in two parts. 
One part of the spectrum is described in terms of the solutions  of the transformed homogeneous Baxter equation which should 
be trigonometric polynomials of degree $2\mathsf{M}^{\prime }$, with $\mathsf{M}^{\prime }=%
\mathsf{N}-1-\mathsf{M}$ and the second part by the inhomogeneous Baxter equation. The comparison with the numerical results then suggests that, at least in the limit of
homogeneous chains, the
 part of the spectrum generated by the trigonometric polynomial solutions of degree 2%
\textsf{N} of the inhomogeneous Baxter equation \rf{inhomogen-Bax-eq-M} coincides with the part
generated by the trigonometric polynomial solutions of degree $2\mathsf{M}^{\prime }$
of the transformed homogeneous Baxter equation.

Finally, in the case $\mathsf{M}>\mathsf{N}$ we have a complete characterization of
the spectrum given by an inhomogeneous Baxter functional equation however
nothing prevent to consider solutions to the homogeneous Baxter equation
once we take the appropriate $Q$-function with $\mathsf{M}>\mathsf{N}$ Bethe roots.
The numerical results  however seem to suggest that considering the
homogeneous Baxter equations is not the proper thing to do in the
homogeneous  limit.

The previous analysis seems to support the idea that in the limit of
homogeneous chain our complete characterization still describe the complete
spectrum of the homogeneous transfer matrix.

\subsection{Comparison with numerical results for the XXX chain}

In the case of the open spin 1/2 XXX chain an ansatz based on  two 
$Q$-functions and an inhomogeneous Baxter functional equation has been first
introduced in \cite{CaoYSW13-2}, the completeness of the spectrum
obtained  by that ansatz has been later verified  numerically
for   small chains   \cite{CaoJYW2013}. Using these
results Nepomechie has introduced a simpler ansatz  and developed some further
numerical analysis in \cite{Nep-2013} confirming once again that the ansatz
defines the complete spectrum for small  chains. Here, we would like to
point out that our complete description of the transfer matrix spectrum
in terms of a inhomogeneous Baxter functional equation  obtained  for the inhomogeneous chains 
has the following well defined homogeneous limit: 
\begin{equation}
\tau (\lambda )Q(\lambda )=\mathbf{A}(\lambda )Q(\lambda -\eta )+\mathbf{A}%
(-\lambda )Q(\lambda +\eta )+F(\lambda )
\end{equation}%
where:%
\begin{eqnarray}
F(\lambda ) &=&8(1-\sqrt{(1+\xi ^{2})})\left( \lambda ^{2}-\left( \eta
/2\right) ^{2}\right) ^{2\mathsf{N}+1}, \\
\mathbf{A}(\lambda ) &=&(-1)^{\mathsf{N}}\frac{2\lambda +\eta }{2\lambda }%
\left(\vphantom{\sqrt{(1+\xi ^{2})}}\lambda -\eta /2+p\right)\left(\sqrt{(1+\xi ^{2})}(\lambda -\eta /2)+q\right)\left(\lambda ^{2}-\left( \eta /2\right) ^{2}\right)^{%
\mathsf{N}}.
\end{eqnarray}%
Taking into account the shift in our
definition of the monodromy matrix which insures that the transfer matrix is
an even function of the spectral parameter, the limit of our inhomogeneous
Baxter functional equation coincides with the ansatz proposed by Nepomechie
in \cite{Nep-2013}. Then the numerical evidences of completeness derived by
Nepomechie in \cite{Nep-2013}  suggest that the exact and complete
characterization that we get for the inhomogeneous chain is still valid and complete in the homogeneous limit.

\section*{Conclusion and outlook}

In this paper we have shown that the transfer matrix spectrum associated to
the most general spin-1/2 representations of the 6-vertex reflection
algebras (rational and trigonometric), on general inhomogeneous chains
is completely characterized in terms of a second order difference functional
equations of Baxter $T$-$Q$ type with an inhomogeneous term depending only
on the inhomogeneities of the chain and the boundary parameters. This
functional $T$-$Q$ equation has been shown to be  equivalent to
the SOV complete characterization of the spectrum when the $Q$-functions
belong to a well defined set of polynomials. The polynomial character of the $Q$%
-function is a central feature of our characterization which allows to
introduce an equivalent finite system of generalized Bethe ansatz equations. Moreover, we have explicitly proven that our
functional characterization holds for all the values of the boundary
parameters for which SOV works, clearly identifying the only 3-dimensional
hyperplanes in the 6-dimensional space of the boundary parameters where our description 
cannot be applied.
 We have also clearly identified the 5-dimensional hyperplanes in
the space of the boundary parameters  where the spectrum (or a part of the spectrum)  can 
be characterized in terms of a homogeneous $T$-$Q$ equation and the polynomial character of the $Q$-functions is then equivalent to a
standard system of Bethe equations. Completeness of this description
 is a built in feature due to the equivalence to the SOV
characterization. 

 The equivalence between our functional $T$-$Q$
equation and the SOV characterization holds  for generic values of the $%
\xi _{a}$ in the $\mathsf{N}$-dimensional space of the inhomogeneity
parameters however there exist hyperplanes for which the conditions \rf{xi-conditions} are not satisfied and so a direct application of the SOV
approach is not possible (at least for the separate variables described in 
\cite{Fald-KN13}) and the limit of homogeneous chains ($\xi _{a}\rightarrow 0$
$\forall a\in \{1,...,\mathsf{N}\}$) clearly belong to these hyperplanes.
From the analyticity of the transfer matrix eigenvalues, of the coefficients
of the functional $T$-$Q$ equation and of the inhomogeneous term in it
w.r.t. the inhomogeneity parameters it is possible to argue that these
functional equations still describes transfer matrix eigenvalues on the
hyperplanes where SOV method cannot be applied and, in particular, in the
homogeneous limit. However, in all these cases the statements about the
simplicity of the transfer matrix spectrum and the completeness of the
 description by our functional $T$-$Q$ equation are not anymore
granted and they require independent proofs. These fundamental issues will
be addressed in a future publication. Here we want just to recall that the
comparison with the few existing numerical results on the subject seems to
suggests that the statement of completeness should be satisfied even in the
homogeneous limit of special interest as it allows to reproduce the spectrum
of the Hamiltonian of the spin-1/2 open XXZ quantum chains under the most
general integrable boundary conditions. 

Finally, it is important to note
that the form of the  Baxter  functional equation for the most general
spin-1/2 representations of the 6-vertex reflection algebras and in
particular the necessity of an inhomogeneous term are mainly imposed
by the requirement that the set of solutions is restricted to
polynomials. Then the problem to get homogeneous Baxter  equations
relaxing this last requirement remains an interesting open problem.

\section*{Acknowledgements}

The authors would like to thank E. Sklyanin and V. Terras for discussions.
 J.M.M. and G. N. are supported by CNRS. 
N.K and J.M.M.  are supported by ANR grant ``DIADEMS''.
 N. K. would like to thank LPTHE, University Paris VI and Laboratoire de Physique, ENS-Lyon for hospitality.


\begin{thebibliography}{99}
\bibitem{Bas06} P. Baseilhac. \newblock The $q$-deformed analogue of the
Onsager algebra: Beyond the Bethe ansatz approach \newblock {\em Nucl. Phys.
B} 754 (2006) 309.

\bibitem{BasK05a} P. Baseilhac, K. Koizumi. \newblock A deformed analogue of
Onsager's symmetry in the XXZ open spin chain \newblock {\em J. Stat. Mech.}
(2005) P10005.

\bibitem{Bax72} R.~Baxter. \newblock Partition function of the eight-vertex
lattice model. \newblock {\em Ann. Phys.}, 70:193--228, 1972.

\bibitem{Bax72a} R.~J. Baxter. \newblock One-dimensional anisotropic {H}%
eisenberg chain. \newblock {\em Ann. Phys.}, 70:323--37, 1972.

\bibitem{Bet31} H.~Bethe. \newblock Z\"ur {T}heorie der {M}etalle {I}. {E}%
igenwerte und {E}igenfunktionen  {A}tomkete. \newblock {\em Zeitschrift
f\"{u}r Physik}, 71:205--226, 1931.


\bibitem{CaoLSW03} J.~Cao, H.-Q. Lin, K.-J. Shi, and Y.~Wang. \newblock %
Exact solution of {$XXZ$} spin chain with unparallel boundary fields. %
\newblock {\em Nuclear Phys. B}, 663(3):487--519, 2003.

\bibitem{CaoYSW13-4} J.~Cao, W.~Yang, K.~Shi, and Y.~Wang. \newblock %
Off-diagonal Bethe ansatz solutions of the anisotropic spin-1/2 chains with
arbitrary boundary fields. \newblock {\em Nuclear Phys. B}, 887:152-175,
2013.

\bibitem{CaoYSW13-1} J.~Cao, W.~Yang, K.~Shi, and Y.~Wang. \newblock %
Off-diagonal Bethe ansatz and exact solution a topological spin ring. %
\newblock  {\em Phys. Rev. Lett.} 111, 137201, 2013.

\bibitem{CaoYSW13-2} J.~Cao, W.~Yang, K.~Shi, and Y.~Wang. \newblock %
Off-diagonal bethe ansatz solution of the XXX spin-chain with  arbitrary
boundary conditions. \newblock {\em Nuclear Phys. B}, 875:152--165, 2013.

\bibitem{CaoYSW13-3} J.~Cao, W.~Yang, K.~Shi, and Y.~Wang. \newblock %
Spin-1/2 XYZ model revisit: general solutions via off-diagonal Bethe ansatz. %
\newblock arXiv preprint arXiv:1307.0280, 2013.

\bibitem{Che84} I.~V. Cherednik. \newblock Factorizing particles on a
half-line and root systems. \newblock {\em Theor. Math. Phys.}, 61:977--983,
1984.

\bibitem{CraRS10} N. Cramp\'e, E. Ragoucy, D. Simon. \newblock Eigenvectors
of open XXZ and ASEP models for a class of non-diagonal boundary conditions %
\newblock {\em J. Stat. Mech.} P11038 (2010).

\bibitem{CraR12} N. Cramp\'e, E. Ragoucy \newblock Generalized coordinate
Bethe ansatz for non-diagonal boundaries \newblock {\em Nucl. Phys. B} 858
(2012) 502.

\bibitem{EssD05} J.~de~Gier and F.~H.~L. Essler. \newblock Bethe ansatz
solution of the asymmetric exclusion process with open  boundaries. %
\newblock {\em Phys. Rev. Lett.}, 95(24):240601, 4, 2005.

\bibitem{EssD06} J.~de~Gier and F.~H.~L. Essler. \newblock Exact spectral
gaps of the asymmetric exclusion process with open  boundaries. \newblock
{\em Journal of Statistical Mechanics: Theory and Experiment}, 
2006(12):P12011, 2006.

\bibitem{Gal08} W. Galleas. \newblock Functional relations from the
Yang-Baxter algebra: Eigenvalues of the XXZ model with non-diagonal twisted
and open boundary conditions \newblock {\em Nucl. Phys. B} 790 (2008) 524.

\bibitem{deGP04} J. de Gier and P. Pyatov. \newblock Bethe ansatz for the
Temperley-Lieb loop model with open boundaries \newblock  {\em J. Stat.
Mech. P03002} (2004).

\bibitem{GroMN12} N.~Grosjean, J.~M. Maillet, and G.~Niccoli. \newblock On
the form factors of local operators in the lattice sine-Gordon
model. \newblock {\em Journal of Statistical Mechanics: Theory and
Experiment}, P10006, 2012.

\bibitem{GroMN13} N. Grosjean, J.-M. Maillet, G. Niccoli, \newblock On the
form factors of local operators in the Bazhanov-Stroganov and chiral Potts
models. \newblock arXiv:1309.4701.

\bibitem{GroN12} N. Grosjean, G. Niccoli, The $\tau_2$-model and the chiral
Potts model revisited: completeness of Bethe equations from Sklyanin's SOV
method. \newblock {\em J. Stat. Mech.} P11005 (2012).

\bibitem{FadST79} L.~D. Faddeev, E.~K. Sklyanin, and L.~A. Takhtajan. %
\newblock Quantum inverse problem method {I}. \newblock {\em Theor. Math.
Phys.}, 40:688--706, 1979.

\bibitem{Fald-KN13} S. Faldella, N. Kitanine, G. Niccoli. \newblock Complete
spectrum and scalar products for the open spin-1/2 XXZ quantum chains with
non-diagonal boundary terms. \newblock Accepted for publication on \emph{J.
Stat. Mech.: Theory Exp.} \newblock arXiv:1307.3960.


\bibitem{FaldN13} S.~Faldella, G.~Niccoli. \newblock SOV approach for
integrable quantum models associated to the most general representations on
spin-1/2 chains of the 8-vertex reflection algebra. \newblock  %
arXiv:1307.5531.

\bibitem{FilK11} G.~Filali and N.~Kitanine. \newblock Spin chains with
non-diagonal boundaries and trigonometric {SOS}  model with reflecting end. %
\newblock {\em SIGMA Symmetry Integrability Geom. Methods Appl.}, 7:Paper
012,  22, 2011.

\bibitem{FraSW08} H. Frahm, A. Seel, T. Wirth. \newblock Separation of
variables in the open XXX chain. \newblock {\em Nucl. Phys. B} {\bf 802}
(2008) 351.

\bibitem{FraGSW11} H. Frahm, J. H. Grelik, A. Seel, T. Wirth. \newblock %
Functional Bethe ansatz methods for the open XXX chain. \newblock {\em J.
Phys. A}  {\bf 44} (2011) 015001.

\bibitem{CaoJYW2013} Y. Jiang, S. Cui, J. Cao, Wen-Li Yang and Y. Wang, %
\newblock Completeness and Bethe root distribution of the spin- 1/2
Heisenberg chain with arbitrary boundary fields. \newblock arXiv:1309.6456v1.


\bibitem{KitKMNST07} N.~Kitanine, K.~Kozlowski, J.~Maillet, G.~Niccoli,
N.~Slavnov, and V.~Terras. \newblock On correlation functions of the open {$%
XXZ$} chain {I}. \newblock {\em J. Stat. Mech.: Theory Exp.}, pages P10009,
37 pp. (electronic),  2007.

\bibitem{KitKMNST08} N.~Kitanine, K.~Kozlowski, J.~Maillet, G.~Niccoli,
N.~Slavnov, and V.~Terras. \newblock On correlation functions of the open {$%
XXZ$} chain {II}. \newblock {\em J. Stat. Mech.: Theory Exp.}, page P07010,
2008. \newblock arXiv:0803.3305.

%

\bibitem{KitMT99} N.~Kitanine, J.~M. Maillet, and V.~Terras. \newblock Form
factors of the {$XXZ$} {H}eisenberg spin-1/2 finite chain. \newblock {\em
Nucl. Phys. B}, 554 [FS]:647--678, 1999.

\bibitem{KitMT00} N.~Kitanine, J.~M. Maillet, and V.~Terras. \newblock %
Correlation functions of the {$XXZ$} heisenberg spin-1/2 chain in a 
magnetic field. \newblock {\em Nucl. Phys. B}, 567 [FS]:554--582, 2000.

\bibitem{MaiT00} J.~M. Maillet and V.~Terras. \newblock On the quantum
inverse scattering problem. \newblock {\em Nucl. Phys. B}, 575:627, 2000. %
\newblock Preprint LPENSL-TH-19/99, hep-th/9911030.

\bibitem{MurN05} R. Murgan and R. I. Nepomechie. \newblock Bethe ansatz
derived from the functional relations of the open XXZ chain for new special
cases.  \newblock {\em J. Stat. Mech.}, (2005) P08002.

\bibitem{Nep02} R.~I. Nepomechie. \newblock Solving the open {XXZ} spin
chain with nondiagonal boundary terms at  roots of unity. \newblock {\em
Nuclear Physics B}, 622(3):615 -- 632, 2002.

\bibitem{Nep04} R.~I. Nepomechie. \newblock Bethe ansatz solution of the
open {$XXZ$} chain with nondiagonal  boundary terms. \newblock {\em J. Phys.
A}, 37(2):433--440, 2004. \newblock Special issue on recent advances in the
theory of quantum integrable  systems.

\bibitem{Nep-2013} R. I. Nepomechie, \newblock Inhomogeneous T-Q equation
for the open XXX chain with general boundary terms: completeness and
arbitrary spin. \newblock arXiv:1307.5049.

\bibitem{Nep-R-2003} R. I. Nepomechie and F. Ravanini, \newblock %
Completeness of the Bethe Ansatz solution of the open XXZ chain with
nondiagonal boundary terms \newblock {\em J. Phys. A 36}, 11391-11402, 2003.

\bibitem{Nep-R-2003add} R. I. Nepomechie and F. Ravanini, \newblock %
Addendum to `Completeness of the Bethe Ansatz solution of the open XXZ chain with
nondiagonal boundary terms' \newblock {\em J. Phys. A 37}, 1945-1946, 2004.

\bibitem{Nep-W-2013} R. I. Nepomechie and C. Wang, \newblock Boundary energy
of the open XXX chain with a non-diagonal boundary term. \newblock  %
arXiv:1310.6305.

\bibitem{N10-1} G. Niccoli. \newblock Reconstruction of Baxter Q-operator
from Sklyanin SOV for cyclic representations of integrable quantum models. %
\newblock {\em Nuclear Phys. B}, 835: 263-283, 2010.

\bibitem{N10-2} G. Niccoli. \newblock Completeness of Bethe Ansatz by
Sklyanin SOV for Cyclic Representations of Integrable Quantum Models. %
\newblock  {\em JHEP}, 1103:123, 2011.

\bibitem{Nic12b} G.~Niccoli. \newblock Non-diagonal open spin 1/2 {X}{X}{Z} quantum chains by separation of  variables: complete spectrum and
matrix elements of some quasi-local  operators. \newblock {\em J.  Stat.
Mech.: Theory and Exp.},  2012(10):P10025, 2012.

\bibitem{N13-1} G. Niccoli, \newblock On the form factors of local operators
in the Bazhanov-Stroganov and chiral Potts models. \newblock ICMP12
Proceedings by World Scientific. \newblock arXiv:1301.4924.

\bibitem{Nic13a} G.~Niccoli. \newblock Antiperiodic spin-1/2 XXZ quantum
chains by separation of  variables: Complete spectrum and form factors. %
\newblock {\em Nucl.Phys. B}, 870: 397 -- 420, 2013.

\bibitem{Nic13b} G.~Niccoli. \newblock Form factors and complete spectrum of
XXX antiperiodic higher spin chains by quantum separation of variables. %
\newblock {\em J. Math. Phys.} 54, 053516 (2013).

\bibitem{Nic13c} G.~Niccoli. \newblock Antiperiodic dynamical 6-vertex model
I: Complete spectrum by SOV, matrix elements of the identity on separate
states and connections to the periodic 8-vertex model. \newblock {\em J.
Phys. A: Math. Theor.} 46 075003, 2013.


\bibitem{NicT10} G. Niccoli and J. Teschner. \newblock The sine-%
Gordon model revisited: I. \newblock {\em  J.  Stat. Mech.: Theory  Exp.},
P09014 (2010).

\bibitem{NicRd05} A. Nichols, V. Rittenberg and J. de Gier. \newblock %
One-boundary Temperley-Lieb algebras in the XXZ and loop models. \newblock
{\em Stat. Mech.} P03003 (2005).

\bibitem{Pro11} T.~Prosen. \newblock Open {X}{X}{Z} spin chain:
Nonequilibrium steady state and a strict  bound on ballistic transport. %
\newblock {\em Phys. Rev. Lett.}, 106:217206, May 2011.

\bibitem{SirPA09} J.~Sirker, R.~G. Pereira, and I.~Affleck. \newblock %
Diffusion and ballistic transport in one-dimensional quantum systems. %
\newblock {\em Phys. Rev. Lett.}, 103:216602, Nov 2009.

\bibitem{Skl85} E.~K. Sklyanin. \newblock The quantum {T}oda chain. %
\newblock In \emph{Nonlinear equations in classical and quantum field theory
({M}eudon/{P}aris, 1983/1984)}, volume 226 of \emph{Lecture Notes in Phys.},
pages 196--233. Springer, Berlin, 1985.

\bibitem{Skl88} E.~Sklyanin. \newblock Boundary conditions for integrable
quantum systems. \newblock {\em J. Phys. A : Math. Gen.}, 21:2375--2389,
1988.

\bibitem{Skl89a} E.~K.~Sklyanin. \newblock Poisson structure of a periodic classical XYZ chain.
{\em J. Soviet Math.}, 1989, v.46, n.1., p.1664-1683.

\bibitem{Skl89b} E.~K.~Sklyanin. \newblock Poisson structure of classical XXZ chain. {\em J. Soviet Math.}, 1989, v.46, n.5, p.2104-2111.

\bibitem{Skl92} E.~K. Sklyanin. \newblock Quantum inverse scattering method.
selected topics. \newblock In M.-L. Ge, editor, \emph{Quantum group and
Quantum Integrable  Systems}, pages 63--97. Nankai Lectures in Mathematical
Physics, World  Scientific, 1992.

%

\bibitem{FadT79} L.~A. Takhtajan and L.~D. Faddeev. \newblock The quantum
method of the inverse problem and the Heisenberg {XYZ}  model. \newblock
{\em Russ. Math. Surveys}, 34(5):11--68, 1979.

\bibitem{YanZ07} W.-L. Yang and Y.-Z. Zhang. \newblock On the second
reference state and complete eigenstates of the open  {$XXZ$} chain. %
\newblock {\em J. High Energy Phys.}, pages 044, 11 pp. (electronic), 2007.









\end{thebibliography}
\end{document}